\documentclass[draftcls,onecolumn]{IEEEtran}
\usepackage{amsmath}
\usepackage{amssymb}
\usepackage{amsfonts}
\usepackage{slashbox}
\usepackage{graphicx}
\usepackage{cite, amsmath, amsfonts, amssymb, psfrag, epsfig,tikz}
\usetikzlibrary{shapes,shadows,arrows}
\newtheorem{definition}{{Definition}}
\newtheorem{theorem}{{Theorem}}
\newtheorem{lemma}{{Lemma}}

\newtheorem{corollary}{{Corollary}}

\DeclareMathAlphabet{\mathpzc}{OT1}{pzc}{m}{it}

\begin{document}

\tikzstyle{line} = [draw, -stealth,semithick]
\tikzstyle{block} = [draw, rectangle, text width=7em, text centered, minimum height=12mm, node distance=8em,semithick]
\tikzstyle{chblock} = [draw, rectangle, text width=4em, text centered, minimum height=40mm, node distance=8em,semithick]

\title{Outer Bounds on the Admissible Source Region for
Broadcast Channels with Correlated Sources}


\author{Kia Khezeli and Jun Chen}


%


\maketitle

\begin{abstract}
Two outer bounds on the admissible source region for
broadcast channels with correlated sources are presented: the first one is strictly tighter than the existing outer bound by Gohari and Anantharam while the second one provides a complete characterization of the admissible source region in the case where the two sources are conditionally independent given the common part. These outer bounds are deduced from the general necessary conditions established for the lossy source broadcast problem via suitable comparisons between the virtual broadcast channel (induced by the source and the reconstructions) and the physical broadcast channel.
\end{abstract}



%

\section{Introduction}\label{sec:introduction}

Let $\{S(t)\}_{t=1}^\infty$ be an i.i.d. random process with marginal distribution $p_S$ over alphabet $\mathcal{S}$. In the lossy source broadcast problem (see Fig. \ref{fig:lossy}), an encoding function $f^{(m,\rho m)}:\mathcal{S}^m\rightarrow\mathcal{X}^{\rho m}$ maps a source block $S^m\triangleq(S(1),\cdots,S(m))$  to a channel input block $X^{\rho m}\triangleq(X(1),\cdots,X(\rho m))$, which is sent over a discrete memoryless broadcast channel $p_{Y_1,Y_2|X}$ with input alphabet $\mathcal{X}$ and output alphabets $\mathcal{Y}_i$, $i=1,2$; at receiver $i$, a decoding function $g^{(\rho m,m)}_i:\mathcal{Y}^{\rho m}_i\rightarrow\hat{\mathcal{S}}^m_i$ maps the channel output block $Y^{\rho m}_i\triangleq(Y_i(1),\cdots,Y_i(\rho m))$ (generated by $X^{\rho m}$) to a source reconstruction block $\hat{S}^m_i\triangleq(\hat{S}_i(1),\cdots,\hat{S}_i(m))$, $i=1,2$. The number of channel uses per source sample, i.e., $\rho$, is referred to as the bandwidth expansion ratio. We assume that  $\mathcal{S}$, $\hat{\mathcal{S}}_1$, $\hat{\mathcal{S}}_2$, $\mathcal{X}$, $\mathcal{Y}_1$, and $\mathcal{Y}_2$ are finite sets throughout the paper.

\begin{definition}\label{def:systemA2}
 Let $w_i:\mathcal{S}\times\hat{\mathcal{S}}_i\rightarrow[0,\infty)$, $i=1,2$, be two distortion measures. We say distortion pair $(d_1,d_2)$ is achievable under distortion measures $w_1$ and $w_2$ subject to bandwidth expansion constraint $\kappa$ if, for every $\epsilon>0$, there exist encoding function $f^{(m,\rho m)}:\mathcal{S}^m\rightarrow\mathcal{X}^{\rho m}$ and decoding functions $g^{(\rho m,m)}_i:\mathcal{Y}^{\rho m}_i\rightarrow\hat{\mathcal{S}}^m_i$, $i=1,2$, with $\rho\leq\kappa+\epsilon$, such that
\begin{align}
\frac{1}{m}\sum\limits_{t=1}^m\mathbb{E}[w_i(S(t),\hat{S}_i(t))]\leq d_i+\epsilon,\quad i=1,2.\nonumber
\end{align}
\end{definition}

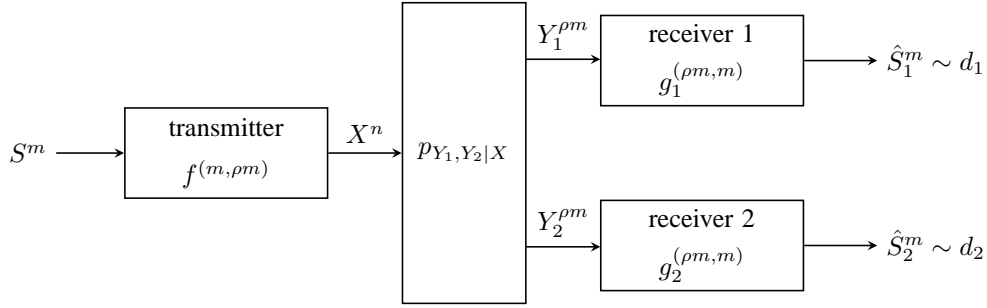
\begin{figure}
\centering
\begin{tikzpicture}
\node [] (Source) {$S^m$};
\node [block, right of=Source , xshift=-0.5em] (Enc) {transmitter\\ $f^{(m,\rho m)}$};
\node [chblock, right of=Enc , xshift=1em] (Channel) {$p_{Y_1,Y_2|X}$};\
\node [block, right of=Channel , xshift=1em, yshift=3.5em] (Dec1) {receiver 1\\ $g_1^{(\rho m,m)}$};
\node [block, right of=Channel , xshift=1em, yshift=-3.5em] (Dec2) {receiver 2\\ $g_2^{(\rho m,m)}$};
\node [right of=Dec1 , xshift=6em] (Rec1) {$\hat{S}_{1}^m\sim d_1$};
\node [right of=Dec2 , xshift=6em] (Rec2) {$\hat{S}_{2}^m\sim d_2$};
\path [line] (18.85em,3.55em)--node[xshift=0em,above] {$Y_{1}^{\rho m}$}(21.65em,3.55em);
\path [line] (18.85em,-3.55em)--node[xshift=0em,above] {$Y_{2}^{\rho m}$}(21.65em,-3.55em);
\path [line] (Source)--(Enc);
\path [line] (Enc) -- node[xshift=0em,above] {$X^n$} (Channel);
\path [line] (Dec1)--(Rec1);
\path [line] (Dec2)--(Rec2);
\end{tikzpicture}
\caption{The lossy broadcast problem}\label{fig:lossy}
\end{figure}



A special case of the lossy source broadcast problem, sometimes referred to as broadcasting correlated sources (see Fig. \ref{fig:correlated}), has received particular attention. In this case, $S(t)=(S_1(t),S_2(t))$ with $S_1(t)$ and $S_2(t)$ jointly distributed according to $p_{(S_1,S_2)}$ over alphabet $\mathcal{S}_1\times\mathcal{S}_2$, $t=1,2,\cdots$, and
receiver $i$ wishes to reconstruct $\{S_i(t)\}_{t=1}^\infty$ almost losslessly, $i=1,2$.

\begin{definition}\label{def:lossless}
A source distribution $p_{(S_1,S_2)}$ is said to be admissible for broadcast channel $p_{Y_1,Y_2|X}$ subject to bandwidth expansion constraint $\kappa$ if, for every $\epsilon>0$, there exist encoding function $f^{(m,\rho m)}:\mathcal{S}^m_1\times\mathcal{S}^m_2\rightarrow\mathcal{X}^n$ and decoding functions $g^{(\rho m,m)}_i:\mathcal{Y}^{\rho m}_i\rightarrow\mathcal{S}^m_i$, $i=1,2$, with $\rho\leq\kappa+\epsilon$, such that
\begin{align*}
\frac{1}{m}\sum\limits_{t=1}^m\mbox{Pr}(S_i(t)\neq\hat{S}_i(t))\leq\epsilon,\quad i=1,2.
\end{align*}
The set of all such $p_{(S_1,S_2)}$ is referred to as the admissible source region for broadcast channel $p_{Y_1,Y_2|X}$ subject to bandwidth expansion constraint $\kappa$.
\end{definition}

Remark: Definition \ref{def:lossless} is a special case of Definition \ref{def:systemA2} with $d_1=d_2=0$ and $w_i:\mathcal{S}\times\mathcal{S}_i\rightarrow\{0,1\}$ given by
\begin{align}
w_i((s_1,s_2),\hat{s}_i)=\left\{
                 \begin{array}{ll}
                   0, & s_i=\hat{s}_i \\
                   1, & \mbox{otherwise}
                 \end{array}
               \right.,\quad i=1,2.\label{eq:measure}
\end{align}
It is worth mentioning that, for the problem of broadcasting correlated sources, typically the more restrictive block error probability constraints are adopted. However, it is clear that outer bounds derived under average symbol error probability constraints automatically hold under block error probability constraints.

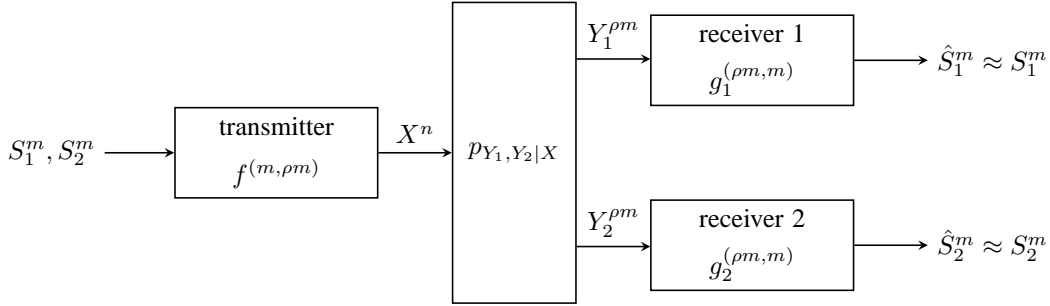
\begin{figure}
\centering
\begin{tikzpicture}
\node [] (Source) {$S^m_1,S^m_2$};
\node [block, right of=Source , xshift=0.5em] (Enc) {transmitter\\ $f^{(m,\rho m)}$};
\node [chblock, right of=Enc , xshift=1em] (Channel) {$p_{Y_1,Y_2|X}$};\
\node [block, right of=Channel , xshift=1em, yshift=3.5em] (Dec1) {receiver 1\\ $g_1^{(\rho m,m)}$};
\node [block, right of=Channel , xshift=1em, yshift=-3.5em] (Dec2) {receiver 2\\ $g_2^{(\rho m,m)}$};
\node [right of=Dec1 , xshift=6.25em] (Rec1) {$\hat{S}_{1}^m\approx S^m_1$};
\node [right of=Dec2 , xshift=6.25em] (Rec2) {$\hat{S}_{2}^m\approx S^m_2$};
\path [line] (19.85em,3.55em)--node[xshift=0em,above] {$Y_{1}^{\rho m}$}(22.65em,3.55em);
\path [line] (19.85em,-3.55em)--node[xshift=0em,above] {$Y_{2}^{\rho m}$}(22.65em,-3.55em);
\path [line] (Source)--(Enc);
\path [line] (Enc) -- node[xshift=0em,above] {$X^n$} (Channel);
\path [line] (Dec1)--(Rec1);
\path [line] (Dec2)--(Rec2);
\end{tikzpicture}
\caption{Broadcasting correlated sources}\label{fig:correlated}
\end{figure}

Han and Costa \cite{HC87} derived an inner bound on the admissible source region; see \cite{KN09} for a minor correction and \cite{MK09} for an alternative characterization. Outer bounds were established by Gohari and Anantharam  and by Kramer, Liang, and Shamai. Note that, due to the lack of cardinality bounds on the auxiliary random variables, neither the original version of the Gohari-Anantharam outer bound \cite{GA08,GA13a} nor the Kramer-Liang-Shamai outer bound \cite{KLS09} is directly computable. Thanks to \cite{Nair11}, a computable characterization of the Gohari-Anantharam outer bound has been found recently \cite{GA13b,KC14}. On the other hand, it is difficult, if not impossible, to express the Kramer-Liang-Shamai outer bound in a computable form because of the fact that certain auxiliary random variables involved in this bound are constrained\footnote{If such a constraint is removed, then the Gohari-Anantharam outer bound is at least as tight as the Kramer-Liang-Shamai outer bound.} to be i.i.d. copies of the source variables. For this reason, only the Gohari-Anantharam outer bound is considered in the present work.

In this paper we establish two necessary conditions for the lossy source broadcast problem. Both conditions are built upon the intuition that the virtual broadcast channel (induced by the source and the reconstructions) is dominated by the physical broadcast channel. Our effort is largely devoted to seeking mathematical formulations that can capture, to a certain extent, this vague intuition. It will be seen that the notion of dominance, which has a precise definition in the point-to-point case due to the source-channel separation theorem, permits several possible generalizations to the broadcast channel setting, and each generalization gives rise to a necessary condition for the lossy source broadcast problem. These necessary conditions, when specialized to the problem of broadcasting correlated sources, yield two outer bounds on the admissible source region: the first one is strictly tighter than the Gohari-Anantharam outer bound while the second one provides a complete characterization of the admissible source region in the case where the two sources are conditionally independent given the common part.

The rest of this paper is organized as follows. We explain our general approach in Section \ref{sec:versus}. Section \ref{sec:Review} contains a short review of the relevant capacity results for broadcast channels. The necessary conditions for the lossy source broadcast problem and the induced outer bounds on the source admissible region are presented in Sections \ref{sec:I} and \ref{sec:II}. We conclude the paper in Section \ref{sec:conclusion}.


\section{Virtual Channel versus Physical Channel}\label{sec:versus}

For the purpose of illustrating our general approach, it is instructive to first consider the point-to-point communication problem. Specifically, in the point-to-point setting, an encoding function $f^{(m,\rho m)}:\mathcal{S}^m\rightarrow\mathcal{X}^{\rho m}$ maps a source block $S^m$  to a channel input block $X^{\rho m}$, which is sent over a discrete memoryless channel $p_{Y|X}$ with input alphabet $\mathcal{X}$ and output alphabet $\mathcal{Y}$; at the receiver end, a decoding function $g^{(\rho m,m)}:\mathcal{Y}^{\rho m}\rightarrow\hat{\mathcal{S}}^m$ maps the channel output block $Y^{\rho m}$ (generated by $X^{\rho m}$) to a source reconstruction block $\hat{S}^m$. For any conditional distribution $p_{\hat{S}^m|S^m}$, let $p_{\hat{S}|S}$ be its single-letterized version defined as
\begin{align*}
p_{\hat{S}|S}(\hat{s}|s)=\frac{1}{m}\sum\limits_{t=1}^mp_{\hat{S}(t)|S(t)}(\hat{s}|s),
\end{align*}
where
\begin{align*}
p_{\hat{S}(t)|S(t)}(\hat{s}|s)=\sum\limits_{\substack{s^m:s(t)=s \\ \hat{s}^m:\hat{s}(t)=\hat{s}}}p_{\hat{S}^m|S^m}(\hat{s}^m|s^m)\prod\limits_{t':t'\neq t}p_S(s(t')).
\end{align*}
One can readily verify that
\begin{align}
\mathbb{E}[w(S,\hat{S})]=\frac{1}{m}\sum\limits_{t=1}^m\mathbb{E}[w(S(t),\hat{S}(t))]\label{eq:predis}
\end{align}
for any distortion measure $w:\mathcal{S}\times\hat{\mathcal{S}}\rightarrow[0,\infty)$. We say that $p_{\hat{S}^m|S^m}$ is degraded with respect to $p_{Y^{\rho m}|X^{\rho m}}$ (where $p_{Y^{\rho m}|X^{\rho m}}(y^{\rho m}|x^{\rho m})=\prod_{q=1}^{\rho m}p_{Y|X}(y(q)|x(q))$) if
\begin{align*}
p_{\hat{S}^m|S^m}(\hat{s}^m|s^m)=\sum\limits_{x^{\rho m},y^{\rho m}}p_{X^{\rho m}|S^m}(x^{\rho m}|s^m)p_{Y^{\rho m}|X^{\rho m}}(y^{\rho m}|x^{\rho m})p_{\hat{S}^m|Y^{\rho m}}(\hat{s}^m|y^{\rho m})
\end{align*}
for some conditional distributions $p_{X^{\rho m}|S^m}$ and $p_{\hat{S}^m|Y^{\rho m}}$; note that, in the point-to-point communication problem, we have
\begin{align*}
&p_{X^{\rho m}|S^m}(x^{\rho m}|s^m)=\mathbb{I}(x^{\rho m}=f^{(m,{\rho m})}(s^m)),\\
&p_{\hat{S}^m|Y^{\rho m}}(\hat{s}^m|y^{\rho m})=\mathbb{I}(\hat{s}^m=g^{({\rho m},m)}(y^{\rho m})),
\end{align*}
where $\mathbb{I}(\cdot)$ is the indicator function.   We shall refer to an arbitrary conditional distribution $p_{\hat{S}|S}$ as a virtual channel, and say that it is realizable through the physical channel $p_{Y|X}$ with bandwidth expansion ratio $\rho$ if it can be obtained, via single-letterization, from certain $p_{\hat{S}^m|S^m}$ degraded\footnote{Since $p_{\hat{S}^m|S^m}$ is only required to be degraded with respect to $p_{Y^{\rho m}|X^{\rho m}}$, we essentially allow non-deterministic encoding and decoding functions.
 However, it can be shown via a standard derandomization argument that restricting encoding and decoding functions to deterministic ones does not affect the set of (asymptotically) realizable virtual channels.} with respect to $p_{Y^{\rho m}|X^{\rho m}}$.  It is worth emphasizing that a realizable $p_{\hat{S}|S}$ is not necessarily degraded with respect to $p_{Y|X}$. Indeed, even in the bandwidth-matched case (i.e., $\rho=1$), there might not exist conditional distributions $p_{X|S}$ and $p_{\hat{S}|Y}$ such that
\begin{align*}
p_{\hat{S}|S}(\hat{s}|s)=\sum\limits_{x,y}p_{X|S}(x|s)p_{Y|X}(y|x)p_{\hat{S}|Y}(\hat{s}|y)
\end{align*}
since otherwise the end-to-end distortion $\frac{1}{m}\sum_{t=1}^m\mathbb{E}[w(S(t),\hat{S}(t))]$ could be achieved\footnote{One could simply use source variable $S$ to generate channel input $X$ via $p_{X|S}$, and use channel output $Y$ to generate reconstruction variable $\hat{S}$ via $p_{\hat{S}|Y}$. In light of (\ref{eq:predis}), the resulting distortion $\mathbb{E}[w(S,\hat{S})]$ is the same as $\frac{1}{m}\sum_{t=1}^m\mathbb{E}[w(S(t),\hat{S}(t))]$.} without coding. Nevertheless, it can be shown that every realizable $p_{\hat{S}|S}$ is dominated by $p_{Y|X}$ in the sense that
\begin{align}
I(S;\hat{S})\leq\rho I(X;Y)\label{eq:ptpdom}
\end{align}
for some input distribution $p_X$. In a certain sense, (\ref{eq:ptpdom}) is the only connection between a generic realizable virtual channel and the physical channel. Indeed, the source-channel separation theorem essentially asserts that a virtual channel $p_{\hat{S}|S}$ is (asymptotically) realizable through the physical channel $p_{Y|X}$ with bandwidth expansion ratio $\rho$ if and only if $p_{\hat{S}|S}$ is dominated by $p_{Y|X}$ in the sense of (\ref{eq:ptpdom}). Note that, given bandwidth expansion ratio $\rho$, the achievability of end-to-end distortion $d$ is equivalent to the existence of a realizable virtual channel $p_{\hat{S}|S}$ satisfying $\mathbb{E}[w(S,\hat{S})]\leq d$. More generally, one can formulate the question of determining whether there exists a realizable virtual channel in a prescribed set; imposing distortion constraints can be viewed as a specific way of choosing such sets.

This perspective can also be adopted in the broadcast channel setting. For any conditional distribution $p_{\hat{S}^m_{1},\hat{S}^m_{2}|S^m}$, we define its single-letterized version $p_{\hat{S}_1,\hat{S}_2|S}$ as
\begin{align*}
p_{\hat{S}_1,\hat{S}_2|S}(\hat{s}_1,\hat{s}_2|s)=\frac{1}{m}\sum\limits_{t=1}^mp_{\hat{S}_1(t),\hat{S}_2(t)|S(t)}(\hat{s}_1,\hat{s}_2|s),
\end{align*}
where
\begin{align*}
p_{\hat{S}_1(t),\hat{S}_2(t)|S(t)}(\hat{s}_1,\hat{s}_2|s)=\sum\limits_{\substack{s^m:s(t)=s \\ \hat{s}^m_i:\hat{s}_i(t)=\hat{s}_i,\mbox{ } i=1,2}}p_{\hat{S}^m_1,\hat{S}^m_2|S^m}(\hat{s}^m_1,\hat{s}^m_2|s^m)\prod\limits_{t':t'\neq t}p_S(s(t')).
\end{align*}
It can be verified that
\begin{align*}
\mathbb{E}[w_i(S,\hat{S}_i)]=\frac{1}{m}\sum\limits_{t=1}^m\mathbb{E}[w_i(S(t),\hat{S}_i(t))]
\end{align*}
for any distortion measure $w_i:\mathcal{S}\times\hat{S}\rightarrow[0,\infty)$, $i=1,2$.
We say that  $p_{\hat{S}^m_{1},\hat{S}^m_{2}|S^m}$ is degraded with respect to $p_{Y^{\rho m}_1,Y^{\rho m}_2|X^{\rho m}}$ (where $p_{Y^{\rho m}_1,Y^{\rho m}_2|X^{\rho m}}(y^{\rho m}_1,y^{\rho m}_2|x^{\rho m})=\prod_{q=1}^{\rho m}p_{Y|X}(y_1(q),y_2(q)|x(q))$) if
\begin{align}
p_{\hat{S}^m_1,\hat{S}^m_2|S^m}(\hat{s}^m_1,\hat{s}^m_2|s^m)=\sum\limits_{x^{\rho m},y^{\rho m}_1,y^{\rho m}_2}p_{X^{\rho m}|S^m}(x^{\rho m}|s^m)p_{Y^{\rho m}_1,Y^{\rho m}_2|X^{\rho m}}(y^{\rho m}_1,y^{\rho m}_2|x^{\rho m})\prod\limits_{i=1}^2p_{\hat{S}^m_i|Y^{\rho m}_i}(\hat{s}^m_i|y^{\rho m}_i)\label{eq:multicond}
\end{align}
for some conditional distributions $p_{X^{\rho m}|S^m}$ and $p_{\hat{S}^m_i|Y^{\rho m}_i}$, $i=1,2$; note that
\begin{align*}
&p_{X^{\rho m}|S^m}(x^{\rho m}|s^m)=\mathbb{I}(x^{\rho m}=f^{(m,{\rho m})}(s^m)),\\
&p_{\hat{S}^m_i|Y^{\rho m}_i}(\hat{s}^m_i|y^{\rho m}_i)=\mathbb{I}(\hat{s}^m_i=g^{({\rho m},m)}_i(y^{\rho m}_i)),\quad i=1,2,
\end{align*}
in the lossy source broadcast problem.
We shall refer to an arbitrary conditional distribution $p_{\hat{S}_1,\hat{S}_2|S}$ as a virtual broadcast channel\footnote{It is worth mentioning that the idea of viewing the conditional distribution of the reconstructions given the source as a virtual broadcast channel was exploited earlier in \cite{TDS11a,CTDS14} through a different angle.}, and say that it is realizable through the physical broadcast channel $p_{Y_1,Y_2|X}$ with bandwidth expansion ratio $\rho$ if it can be obtained, via single-letterization, from certain $p_{\hat{S}^m_1,\hat{S}^m_2|S^m}$ degraded with respect to $p_{Y^{\rho m}_1,Y^{\rho m}_2|X^{\rho m}}$. The fundamental problem here is to determine whether there exists a realizable virtual broadcast channel satisfying the distortion constraints (or more generally, whether there exists a realizable virtual broadcast channel in a prescribed set). It is conceivable that every realizable virtual broadcast channel must be dominated, in a certain sense,  by the physical broadcast channel. However, it is  apparently a formidable task to develop a computable notion of dominance that can completely characterize the set of realizable virtual broadcast channels (since that would solve several long-standing open problems in network information theory). Nevertheless, if one does not insist on such a complete characterization, then it is indeed possible to establish certain connections between a generic realizable virtual broadcast channel and the physical broadcast channel by suitably generalizing (\ref{eq:ptpdom}). For example, it is straightforward to show that, if $p_{\hat{S}|S}$ is realizable through $p_{Y_1,Y_2|X}$ with bandwidth expansion ratio $\rho$, then $p_{\hat{S}|S}$ must be dominated by $p_{Y_1,Y_2|X}$ in the sense that
\begin{align*}
I(S;\hat{S}_i)\leq\rho I(X;Y_i),\quad i=1,2,
\end{align*}
for some input distribution $p_X$. This is by no means the only possible generalization of (\ref{eq:ptpdom}) to the broadcast channel setting, and two stronger notions of dominance will be presented in Sections \ref{sec:I} and \ref{sec:II}. Note that each notion of dominance implicitly provides an outer bound on the set of realizable virtual broadcast channels, which in turn yields a necessary condition for the achievability of any given distortion pair.

\section{Review of Capacity Results for Broadcast Channels}\label{sec:Review}

We shall give a brief review of certain capacity results for broadcast channels that are relevant to the notions of dominance developed in Sections \ref{sec:I} and \ref{sec:II}.
Let $p_{Y_1,Y_2|X}$ be a discrete memoryless broadcast channel  with input alphabet $\mathcal{X}$ and output alphabets $\mathcal{Y}_i$, $i=1,2$. A length-$n$ coding scheme (see Fig. \ref{fig:broadcast}) for $p_{Y_1,Y_2|X}$ consists of
\begin{itemize}
\item a common message $M_0$ and two private messages $M_i$, $i=1,2$, where $(M_0,M_1,M_2)$ is uniformly distributed over $\mathcal{M}_0\times\mathcal{M}_1\times\mathcal{M}_2$,

\item an encoding function $f^{(n)}:\mathcal{M}_0\times\mathcal{M}_1\times\mathcal{M}_2\rightarrow\mathcal{X}^n$ that maps $(M_0,M_1,M_2)$ to a channel input block $X^n_1$,

\item two decoding functions $g^{(n)}_i:\mathcal{Y}^n_i\rightarrow\mathcal{M}_0\times\mathcal{M}_i$, $i=1,2$, where $g^{(n)}_i$ maps the channel output block at receiver $i$, i.e., $Y^n_{i,1}$, to $(\hat{M}_{0i}, \hat{M}_i)$, $i=1,2$.
\end{itemize}

\begin{figure}
\centering
\begin{tikzpicture}
\node [] (Source) {$M_0,M_1,M_2$};
\node [block, right of=Source , xshift=1.5em] (Enc) {transmitter\\ $f^{(n)}$};
\node [chblock, right of=Enc , xshift=1em] (Channel) {$p_{Y_{1},Y_{2}|X}$};\
\node [block, right of=Channel , xshift=1em, yshift=3.5em] (Dec1) {receiver 1\\ $g_1^{(n)}$};
\node [block, right of=Channel , xshift=1em, yshift=-3.5em] (Dec2) {receiver 2\\ $g_2^{(n)}$};
\node [right of=Dec1 , xshift=6em] (Rec1) {$\hat{M}_{01},\hat{M}_1$};
\node [right of=Dec2 , xshift=6em] (Rec2) {$\hat{M}_{02},\hat{M}_2$};
\path [line] (Source)--(Enc);
\path [line] (Enc) -- node[xshift=0em,above] {$X_1^n$} (Channel);
\path [line] (20.85em,3.55em)--node[xshift=0em,above] {$Y_{1}^n$}(23.65em,3.55em);
\path [line] (20.85em,-3.55em)--node[xshift=0em,above] {$Y_{1}^n$}(23.65em,-3.55em);

\path [line] (Dec1)--(Rec1);
\path [line] (Dec2)--(Rec2);
\end{tikzpicture}
\caption{Broadcast channel with common and private messages}\label{fig:broadcast}
\end{figure}
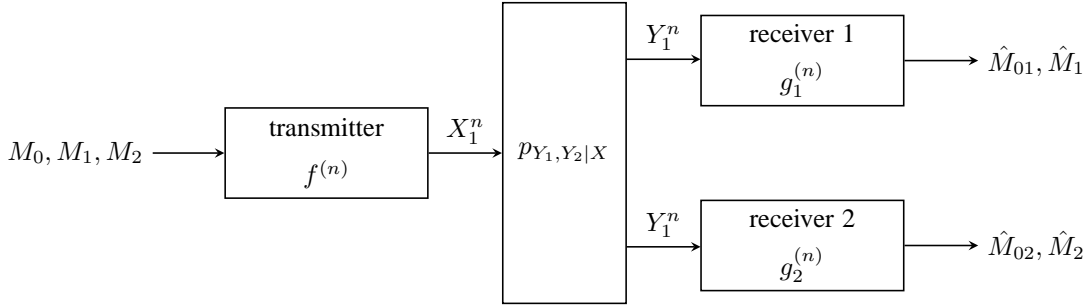

\begin{definition}
A rate triple $(R_0,R_1,R_2)\in\mathbb{R}^3_+$ is said to be achievable for broadcast channel $p_{Y_1,Y_2|X}$ if there exists a sequence of encoding functions $f^{(n)}:\mathcal{M}_0\times\mathcal{M}_1\times\mathcal{M}_2\rightarrow\mathcal{X}^n$ with $\frac{1}{n}\log|\mathcal{M}_i|\geq R_i$, $i=0,1,2$, and decoding functions $g^{(n)}_i:\mathcal{Y}^n_i\rightarrow\mathcal{M}_0\times\mathcal{M}_i$, $i=1,2$, such that
\begin{align*}
\lim\limits_{n\rightarrow\infty}\mbox{Pr}\{(\hat{M}_{01},\hat{M}_1)\neq(M_0,M_1)\mbox{ or }(\hat{M}_{02},\hat{M}_2)\neq(M_0,M_2)\}=0.
\end{align*}
The capacity region $\mathcal{C}(p_{Y_1,Y_2|X})$ is the closure of the set of all achievable $(R_0,R_1,R_2)$ for broadcast channel $p_{Y_1,Y_2|X}$.
\end{definition}




Let $\mathcal{C}_{\textsf{in}}(p_X, p_{Y_1,Y_2|X})$ denote the set of $(R_0,R_1,R_2)\in\mathbb{R}^3_+$ satisfying
\begin{align*}
&R_0\leq\min\{I(V_0;Y_1),I(V_0;Y_2)\},\\
&R_0+R_i\leq I(V_0,V_i;Y_i),\quad i=1,2,\\
&R_0+R_1+R_2\leq\min\{I(V_0;Y_1),I(V_0;Y_2)\}+I(V_1;Y_1|V_0)+I(V_2;Y_2|V_0)-I(V_1;V_2|V_0)
\end{align*}
for some $p_{V_0,V_1,V_2,X,Y_1,Y_2}=p_{V_0,V_1,V_2|X}p_Xp_{Y_1,Y_2|X}$. Here it suffices to consider $|\mathcal{V}_0|\leq|\mathcal{X}|+4$ and $|\mathcal{V}_i|\leq|\mathcal{X}|$, $i=1,2$; moreover, there is no loss of generality in assuming that $X$ is a deterministic function of $(V_0,V_1,V_2)$ \cite[Theorem 1]{GA12}\footnote{It is expected \cite{Gohari13} that  one can further improve the cardinality bounds on $\mathcal{V}_i$, $i=1,2$, to $|\mathcal{V}_1|+|\mathcal{V}_2|\leq|\mathcal{X}|+1$ by leveraging the techniques developed in \cite{AGN13}.}. Define
\begin{align*}
\mathcal{C}_{\textsf{in}}(p_{Y_1,Y_2|X})=\bigcup\limits_{p_X}\mathcal{C}_{\textsf{in}}(p_X, p_{Y_1,Y_2|X}).
\end{align*}
We have \cite[Theorem 1]{GP80} \cite[p. 391, Problem 10(c)]{CK81}
\begin{align*}
\mathcal{C}_{\textsf{in}}(p_{Y_1,Y_2|X})\subseteq\mathcal{C}(p_{Y_1,Y_2|X}).
\end{align*}
Note that $\mathcal{C}_{\textsf{in}}(p_{Y_1,Y_2|X})$ is widely known as Marton's inner bound (see \cite[Theorem 2]{Marton79} for the case $R_0=0$).

Let $\mathcal{C}_{\textsf{out}}(p_X, p_{Y_1,Y_2|X})$ denote the set of $(R_0,R_1,R_2)\in\mathbb{R}^3_+$ satisfying
\begin{align*}
&R_0\leq\min\{I(V_0;Y_1),I(V_0;Y_2)\},\\
&R_0+R_i\leq\min\{I(V_0;Y_1),I(V_0;Y_2)\}+I(V_i;Y_i|V_0),\quad i=1,2,\\
&R_0+R_1+R_2\leq\min\{I(V_0;Y_1),I(V_0;Y_2)\}+I(V_1;Y_1|V_0)+I(X;Y_2|V_0,V_1),\\
&R_0+R_1+R_2\leq\min\{I(V_0;Y_1),I(V_0;Y_2)\}+I(V_2;Y_2|V_0)+I(X;Y_1|V_0,V_2)
\end{align*}
for some $p_{V_0,V_1,V_2,X,Y_1,Y_2}=p_{V_0,V_1,V_2|X}p_Xp_{Y_1,Y_2|X}$. Here it suffices to consider $|\mathcal{V}_0|\leq|\mathcal{X}|+5$ and $|\mathcal{V}_i|\leq|\mathcal{X}|+1$, $i=1,2$ \cite{Nair11}\footnote{In fact, the cardinality bounds on  $\mathcal{V}_i$, $i=1,2$, can be further improved to $|\mathcal{V}_i|\leq |\mathcal{X}|$, $i=1,2$ \cite{Nair13p}.}. Define
\begin{align*}
\mathcal{C}_{\textsf{out}}(p_{Y_1,Y_2|X})=\bigcup\limits_{p_X}\mathcal{C}_{\textsf{out}}(p_X, p_{Y_1,Y_2|X}).
\end{align*}
We have \cite{Nair11}
\begin{align}
\mathcal{C}(p_{Y_1,Y_2|X})\subseteq\mathcal{C}_{\textsf{out}}(p_{Y_1,Y_2|X}).\label{eq:outerbound}
\end{align}

It is worth noting \cite{GP80,Nair13p} that $\mathcal{C}_{\textsf{in}}(p_X, p_{Y_1,Y_2|X})$ can be defined equivalently as the set of $(R_0,R_1,R_2)\in\mathbb{R}^3_+$ satisfying
\begin{align*}
&R_0\leq\min\{I(V_0;Y_1),I(V_0;Y_2)\},\\
&R_0+R_i\leq\min\{I(V_0;Y_1),I(V_0;Y_2)\}+ I(V_i;Y_i|V_0),\quad i=1,2,\\
&R_0+R_1+R_2\leq\min\{I(V_0;Y_1),I(V_0;Y_2)\}+I(V_1;Y_1|V_0)+I(V_2;Y_2|V_0)-I(V_1;V_2|V_0)
\end{align*}
for some $p_{V_0,V_1,V_2,X,Y_1,Y_2}=p_{V_0,V_1,V_2|X}p_Xp_{Y_1,Y_2|X}$. With this equivalent definition of $\mathcal{C}_{\textsf{in}}(p_X, p_{Y_1,Y_2|X})$, one can readily show
\begin{align*}
\mathcal{C}_{\textsf{in}}(p_X,p_{Y_1,Y_2|X})\subseteq\mathcal{C}_{\textsf{out}}(p_X,p_{Y_1,Y_2|X})
\end{align*}
by invoking the fact that, for any $p_{V_0,V_1,V_2,X,Y_1,Y_2}=p_{V_0,V_1,V_2|X}p_Xp_{Y_1,Y_2|X}$,
\begin{align*}
I(X;Y_2|V_0,V_1)&\geq I(V_2;Y_2|V_0,V_1)\\
&=I(V_2;Y_2,V_1|V_0)-I(V_1;V_2|V_0)\\
&\geq I(V_2;Y_2|V_0)-I(V_1;V_2|V_0)
\end{align*}
and similarly
\begin{align*}
I(X;Y_1|V_0,V_2)\geq I(V_1;Y_1|V_0)-I(V_1;V_2|V_0).
\end{align*}

\section{Approach I: Comparison of Certain Measurements Induced by Test Distributions}\label{sec:I}

\subsection{The Lossy Source Broadcast Problem}\label{sec:lossy}


As explained in Section \ref{sec:versus}, our goal is to develop suitable notions of dominance that can (partially) characterize the set of realizable virtual broadcast channels. One such notion is established in the following Lemma (with its proof relegated to Appendix \ref{app:ordering}). Roughly speaking, it shows that every realizable virtual broadcast channel $p_{\hat{S}_1,\hat{S}_2|S}$ (with input distribution $p_S$) is dominated by the physical broadcast channel $p_{Y_1,Y_2|X}$  (with certain input distribution $p_X$) in the sense that, given any test distribution $p_{U_0,\cdots,U_L|S}$ for $p_{\hat{S}_1,\hat{S}_2|S}$, one can find a corresponding test distribution $p_{V_0,\cdots,V_L|X}$ for $p_{Y_1,Y_2|X}$ such that certain measurements based on $p_{U_0,\cdots,U_L|S}p_Sp_{\hat{S}_1,\hat{S}_2|S}$ are less than or equal to those based on $p_{V_0,\cdots,V_L|X}p_Xp_{Y_1,Y_2|X}$ multiplied by bandwidth expansion ratio $\rho$.

\begin{lemma}\label{lem:ordering}
If a virtual broadcast channel $p_{\hat{S}_1,\hat{S}_2|S}$ is realizable through the physical broadcast channel $p_{Y_1,Y_2|X}$ with bandwidth expansion ratio $\rho$, then there exists an input distribution $p_X$ such that, for any $p_{U_0,\cdots,U_L,S,\hat{S}_1,\hat{S}_2}=p_{U_0,\cdots,U_L|S}p_Sp_{\hat{S}_1,\hat{S}_2|S}$, one can find $p_{V_0,\cdots,V_L,X,Y_1,Y_2}=p_{V_0,\cdots,V_L|X}p_Xp_{Y_1,Y_2|X}$ satisfying\footnote{Here $L$ and $k$ are arbitrary positive integers. We define $U_{\mathcal{A}}=(U_i)_{i\in\mathcal{A}}$ when $\mathcal{A}$ is a non-empty subset of $\{0,\cdots,L\}$ and define $U_{\mathcal{A}}=0$ otherwise; $V_{\mathcal{A}}$ is defined analogously.}
\begin{align*}
\sum\limits_{i=1}^kI(U_{\mathcal{A}_i};\hat{S}_{a(i)}|U_{\cup_{j=1}^{i-1}\mathcal{A}_j})\leq\rho\sum\limits_{i=1}^kI(V_{\mathcal{A}_i};Y_{a(i)}|V_{\cup_{j=1}^{i-1}\mathcal{A}_j})
\end{align*}
for any $\mathcal{A}_i\subseteq\{0,\cdots,L\}$ and $a(i)\in\{1,2\}$, $i=1,\cdots,k$.
\end{lemma}


The following result, which gives a general necessary condition for the lossy source broadcast problem, is a simple consequence of Lemma \ref{lem:ordering}. Its proof can be found in Appendix \ref{app:TheoremI}.

\begin{theorem}\label{thm:ordering}
If distortion pair $(d_1,d_2)$ is achievable under distortion measures $w_1$ and $w_2$ subject to bandwidth expansion constraint $\kappa$, then there exists a virtual broadcast channel $p_{\hat{S}_1,\hat{S}_2|S}$ with $\mathbb{E}[w_i(S,\hat{S}_i)]\leq d_i$, $i=1,2$, and an input distribution $p_X$ such that, for any $p_{U_0,U_1,U_2,S,\hat{S}_1,\hat{S}_2}=p_{U_0,U_1,U_2|S}p_Sp_{\hat{S}_1,\hat{S}_2|S}$, one can find  $p_{V_0,V_1,V_2,X,Y_1,Y_2}=p_{V_0,V_1,V_2|X}p_Xp_{Y_1,Y_2|X}$ satisfying
\begin{align*}
&I(U_0;\hat{S}_i)\leq\kappa I(V_0;Y_i), \quad i=1,2, \\
&I(U_0,U_i;\hat{S}_i)\leq\kappa I(V_0,V_i;Y_i),\quad i=1,2, \\
&I(U_0;\hat{S}_1)+I(U_2;\hat{S}_2|U_0)\leq\kappa[I(V_0;Y_1)+I(V_2;Y_2|V_0)], \\
&I(U_0;\hat{S}_2)+I(U_1;\hat{S}_1|U_0)\leq\kappa[I(V_0;Y_2)+I(V_1;Y_1|V_0)], \\
&I(U_0,U_1;\hat{S}_1)+I(S;\hat{S}_2|U_0,U_1)\leq\kappa[I(V_0,V_1;Y_1)+I(X;Y_2|V_0,V_1)], \\
&I(U_0,U_2;\hat{S}_2)+I(S;\hat{S}_1|U_0,U_2)\leq\kappa[I(V_0,V_2;Y_2)+I(X;Y_1|V_0,V_2)], \\
&I(U_0;\hat{S}_1)+I(U_2;\hat{S}_2|U_0)+I(S;\hat{S}_1|U_0,U_2)\leq\kappa[I(V_0;Y_1)+I(V_2;Y_2|V_0)+I(X;Y_1|V_0,V_2)], \\
&I(U_0;\hat{S}_2)+I(U_1;\hat{S}_1|U_0)+I(S;\hat{S}_2|U_0,U_1)\leq\kappa[I(V_0;Y_2)+I(V_1;Y_1|V_0)+I(X;Y_2|V_0,V_1)].
\end{align*}
Here it suffices to consider $|\mathcal{U}_0|\leq|\mathcal{S}|$, $|\mathcal{V}_0|\leq|\mathcal{X}|+5$, $|\mathcal{U}_i|\leq |\mathcal{S}|$, and $|\mathcal{V}_i|\leq |\mathcal{X}|$, $i=1,2$.
\end{theorem}

Remark: Gohari and Anantharam independently obtained a necessary condition for the lossy source broadcast channel \cite[Theorem 2]{GA13b}, which is, roughly speaking, a special case of Theorem \ref{thm:ordering} with $\kappa=1$, $S=(S_1,S_2)$, and $U_i=S_i$, $i=1,2$.

The following result is a direct consequence of Theorem \ref{thm:ordering}.
\begin{corollary}\label{cor:capacitycomp}
If distortion pair $(d_1,d_2)$ is achievable under distortion measures $w_1$ and $w_2$ subject to bandwidth expansion constraint $\kappa$, then there exists a virtual broadcast channel $p_{\hat{S}_1,\hat{S}_2|S}$ with $\mathbb{E}[w_i(S,\hat{S}_i)]\leq d_i$, $i=1,2$ and an input distribution $p_X$ such that
\begin{align*}
\mathcal{C}_{\textsf{out}}(p_S,p_{\hat{S}_1,\hat{S}_2|S})\subseteq\kappa\mathcal{C}_{\textsf{out}}(p_X, p_{Y_1,Y_2|X}).
\end{align*}
\end{corollary}

\subsection{An Improved Outer Bound on the Source Admissible Region}\label{sec:lossless}

The following outer bound on the admissible source region was established by Gohari and Anantharam \cite[Corollary 2]{GA13b} (see also \cite[Theorem 2]{KC14}).



Let $S_0$ denote the common part between $S_1$ and $S_2$ in the sense of \cite{GK73,Witsenhausen75}.
\begin{theorem}\label{thm:lossless}
If $p_{(S_1,S_2)}$ is admissible for broadcast channel $p_{Y_1,Y_2|X}$ subject to bandwidth expansion constraint $\kappa$, then there exists $p_{V_0,V_1,V_2,X,Y_1,Y_2}=p_{V_0,V_1,V_2|X}p_Xp_{Y_1,Y_2|X}$ such that
\begin{align*}
&H(S_0)\leq\kappa\min\{I(V_0;Y_1),I(V_0;Y_2)\},\\
&H(S_i)\leq\kappa[\min\{I(V_0;Y_1),I(V_0;Y_2)\}+I(V_i;Y_i|V_0)],\quad i=1,2,\\
&H(S_1,S_2)\leq\kappa[\min\{I(V_0;Y_1),I(V_0;Y_2)\}+I(V_1;Y_1|V_0)+I(X;Y_2|V_0,V_1)],\\
&H(S_1,S_2)\leq\kappa[\min\{I(V_0;Y_1),I(V_0;Y_2)\}+I(V_2;Y_2|V_0)+I(X;Y_1|V_0,V_2)].
\end{align*}
\end{theorem}


It is easy to observe that the inequalities in the statement of Theorem \ref{thm:lossless} closely resemble those in the definition of $\mathcal{C}_{\textsf{out}}(p_{Y_1,Y_2|X})$. In fact, one can readily establish Theorem 2 by interpreting $S_0$ as the common message and $S_i$ as the message (including both the private message $M_i$ and the common message $M_0$) intended for receiver $i$, $i=1,2$, and then following the proof of (\ref{eq:outerbound}). However, this approach is not completely satisfactory. Note that $M_0$, $M_1$, and $M_2$ are assumed to be independent. If the correspondence between sources and messages is exact, then $S_1\leftrightarrow S_0\leftrightarrow S_2$ must form a Markov chain. That is to say, the source-message correspondence does not fully capture the dependence structure between $S_1$ and $S_2$.

We shall show that one can obtain a tighter outer bound on the admissible source region by specializing Theorem \ref{thm:ordering} to the case of broadcasting correlated sources. Note that, if $\sum_{s,\hat{s}_i}p_S(s)p_{\hat{S}_i|S}(\hat{s}_1|s)w_i(s,\hat{s}_i)=0$ for $w_i$ given by (\ref{eq:measure}), $i=1,2$, then we must have\footnote{More precisely, we have $p_{\hat{S}_1,\hat{S}_2|S}=p_{S_1,S_2|(S_1,S_2)}$ when  the input alphabet restricted to $\{s\in\mathcal{S}: p_{S}(s)>0\}$.} $p_{\hat{S}_1,\hat{S}_2|S}=p_{S_1,S_2|(S_1,S_2)}$, which is a deterministic broadcast channel; moreover, in this case, there is no loss of optimality in choosing $U_i=S_i$, $i=1,2$. As a consequence, we obtain the following outer bound on the admissible source region.

\begin{theorem}\label{thm:losslessnew}
If $p_{(S_1,S_2)}$ is admissible for broadcast channel $p_{Y_1,Y_2|X}$ subject to bandwidth expansion constraint $\kappa$, then there exists an input distribution $p_X$ such that,
for any $p_{U,(S_1,S_2)}=p_{U|(S_1,S_2)}p_{(S_1,S_2)}$, one can find  $p_{V_0,V_1,V_2,X,Y_1,Y_2}=p_{V_0,V_1,V_2|X}p_Xp_{Y_1,Y_2|X}$ satisfying
\begin{align*}
&I(U;S_i)\leq\kappa I(V_0;Y_i), \quad i=1,2,\\
&H(S_i)\leq\kappa I(V_0,V_i;Y_i),\quad i=1,2,\\
&I(U;S_1)+H(S_2|U)\leq\kappa[I(V_0;Y_1)+I(V_2;Y_2|V_0)],\\
&I(U;S_2)+H(S_1|U)\leq\kappa[I(V_0;Y_2)+I(V_1;Y_1|V_0)],\\
&I(U;S_1)+H(S_1,S_2|U)\leq\kappa[I(V_0,V_1;Y_1)+I(X;Y_2|V_0,V_1)],\\
&I(U;S_2)+H(S_1,S_2|U)\leq\kappa[I(V_0,V_2;Y_2)+I(X;Y_1|V_0,V_2)],\\
&I(U;S_1)+H(S_1,S_2|U)\leq\kappa[I(V_0;Y_1)+I(V_2;Y_2|V_0)+I(X;Y_1|V_0,V_2)],\\
&I(U;S_2)+H(S_1,S_2|U)\leq\kappa[I(V_0;Y_2)+I(V_1;Y_1|V_0)+I(X;Y_2|V_0,V_1)].
\end{align*}
Here it suffices to consider $|\mathcal{U}|\leq|\mathcal{S}_1|\times|\mathcal{S}_2|$, $|\mathcal{V}_0|\leq|\mathcal{X}|+5$, and $|\mathcal{V}_i|\leq |\mathcal{X}|$, $i=1,2$.
\end{theorem}

It is easy to verify that
\begin{align}
\mathcal{C}_{\textsf{in}}(p_{(S_1,S_2)},p_{S_1,S_2|(S_1,S_2)})=\mathcal{C}_{\textsf{out}}(p_{(S_1,S_2)},p_{S_1,S_2|(S_1,S_2)})=\mathcal{C}(p_{(S_1,S_2)}),\label{eq:inview}
\end{align}
where $\mathcal{C}(p_{(S_1,S_2)})$ denotes the set of $(R_0,R_1,R_2)\in\mathbb{R}^3_+$ satisfying
\begin{align*}
&R_0\leq\min\{I(U;S_1),I(U;S_2)\},\\
&R_0+R_i\leq\min\{I(U;S_1),I(U;S_2)\}+H(S_i|U),\quad i=1,2,\\
&R_0+R_1+R_2\leq\min\{I(U;S_1),I(U;S_2)\}+H(S_1,S_2|U)
\end{align*}
for some $p_{U,(S_1,S_2)}=p_{U|(S_1,S_2)}p_{(S_1,S_2)}$ with $|\mathcal{U}|\leq|\mathcal{S}_1|\times|\mathcal{S}_2|+2$. One can deduce the following result from Theorem \ref{thm:losslessnew} (or Corollary \ref{cor:capacitycomp}).

\begin{corollary}\label{cor:losslesscor}
If $p_{(S_1,S_2)}$ is admissible for broadcast channel $p_{Y_1,Y_2|X}$ subject to bandwidth expansion constraint $\kappa$, then there exists an input distribution $p_X$ such that
\begin{align}
\mathcal{C}(p_{(S_1,S_2)})\subseteq\kappa\mathcal{C}_{\textsf{out}}(p_X, p_{Y_1,Y_2|X}).\label{eq:implyy1}
\end{align}
\end{corollary}

The necessary condition in Theorem \ref{thm:lossless} can be written compactly as
\begin{align}
(H(S_0),H(S_1|S_0),H(S_2|S_0))\in\kappa\mathcal{C}_{\textsf{out}}(p_{Y_1,Y_2|X})\label{eq:degenerate}
\end{align}
when $S_1\leftrightarrow S_0\leftrightarrow S_2$ form a Markov chain. The following result shows that the same simplification is possible for Theorem \ref{thm:losslessnew} and, as a consequence, these two theorems are equivalent in this special case.

\begin{theorem}\label{thm:degenerate}
The necessary condition in Theorem \ref{thm:losslessnew} is equivalent to (\ref{eq:degenerate}) when $S_1\leftrightarrow S_0\leftrightarrow S_2$ form a Markov chain.
\end{theorem}
\begin{IEEEproof}
See Appendix \ref{app:degenerate}.
\end{IEEEproof}

It is clear that one can recover Theorem \ref{thm:lossless} from Theorem \ref{thm:losslessnew} by choosing $U=S_0$. Therefore, the new outer bound is at least as tight as the Gohari-Anantharam outer bound. We shall give an example to show that the improvement can be strict. Our example is motivated by the observation that the characterization of the capacity region of the deterministic broadcast channel with a common message involves an auxiliary random variable which is not necessarily a function of the channel input \cite{Han81} as well as the observation that $H(S_0)$ is not a continuous function of $p_{(S_1,S_2)}$.

Now consider the example where the physical broadcast channel is the Blackwell channel $p^B_{Y_1,Y_2|X}$, where
\begin{align*}
&p^B_{Y_1,Y_2|X}(y_1,y_2|x)\\
&=\left\{
                           \begin{array}{ll}
                             1, & (x,y_1,y_2)=(0,0,0), (1,1,1), (2,0,1)\\
                             0, & \mbox{otherwise}
                           \end{array}
                         \right.
\end{align*}
with $x\in\{0,1,2\}$ and $y_i\in\{0,1\}$, $i=1,2$; moreover, let $S_1=(\tilde{S}_0(\alpha),\tilde{S}_1)$ and $S_2=(\bar{S}_0(\alpha),\bar{S}_2)$ with $\alpha\in[H^{-1}_b(\frac{1}{2}\log_23-\frac{2}{3}),H^{-1}_b(\log_23-\frac{4}{3}))$, where $\tilde{S}_0(\alpha)$, $\bar{S}_0(\alpha)$, $\tilde{S}_1$,  and $\bar{S}_2$ are binary random variables defined over $\{0,1\}$, and $H^{-1}_b(\cdot):[0,1]\rightarrow[0,\frac{1}{2}]$ is the inverse of the binary entropy function $H_b(\cdot)$.
Specifically, we assume that $(\tilde{S}_0(\alpha), \bar{S}_0(\alpha))$, $\tilde{S}_1$, and $\bar{S}_2$ are mutually independent with
\begin{align*}
&p_{\tilde{S}_1}(0)=p_{\bar{S}_2}(0)=H^{-1}_b(\frac{2}{3}),\\
&p_{\tilde{S}_0(\alpha)}(0)=p_{\bar{S}_0(\alpha)}(0)=\alpha,\\
&p_{\bar{S}_0(\alpha)|\tilde{S}_0(\alpha)}(1|0)=p_{\tilde{S}_0(\alpha)|\bar{S}_0(\alpha)}(1|0)=\beta(\alpha),\\
&p_{\bar{S}_0(\alpha)|\tilde{S}_0(\alpha)}(0|1)=p_{\tilde{S}_0(\alpha)|\bar{S}_0(\alpha)}(0|1)=1-\frac{\alpha\beta(\alpha)}{1-\alpha},
\end{align*}
where $\beta(\alpha)$ is the unique solution in $(0,1-\alpha]$ of the following equation
\begin{align*}
H(\tilde{S}_0(\alpha), \bar{S}_0(\alpha))=\log_23-\frac{4}{3}.
\end{align*}
Note that such $S_1$ and $S_2$ have no non-trivial common part, i.e., $H(S_0)=0$.
By setting $p_X(0)=p_X(1)=p_X(2)=\frac{1}{3}$, $V_0=0$, and $V_i=Y_i$, $i=1,2$, one can readily verify that $p_{(S_1,S_2)}$ satisfies the necessary condition in Theorem \ref{thm:lossless} with $\kappa=1$ for any $\alpha\in[H^{-1}_b(\frac{1}{2}\log_23-\frac{2}{3}),H^{-1}_b(\log_23-\frac{4}{3}))$. However, we shall show that this is not the case for Theorem \ref{thm:losslessnew}. It is easy to see that, if the afore-described $p_{(S_1,S_2)}$ is admissible for the Blackwell channel $p^B_{Y_1,Y_2|X}$ subject to bandwidth expansion constraint $\kappa$, then, by Corollary \ref{cor:losslesscor} as well as the fact that $\mathcal{C}_{\textsf{out}}(p^B_{Y_1,Y_2|X})=\mathcal{C}(p^B_{Y_1,Y_2|X})$, we must have
\begin{align}
\mathcal{C}(p_{(S_1,S_2)})\subseteq\kappa\mathcal{C}(p^B_{Y_1,Y_2|X}).\label{eq:violate}
\end{align}
By choosing $U=(\tilde{S}_0(\alpha), \bar{S}_0(\alpha))$, one can readily verify that $(H_b(\alpha),\frac{2}{3},\frac{2}{3})$ is contained in $\mathcal{C}(p_{(S_1,S_2)})$ for any $\alpha\in[H^{-1}_b(\frac{1}{2}\log_23-\frac{2}{3}),H^{-1}_b(\log_23-\frac{4}{3}))$. On the other hand, it follows from \cite[Lemma 1]{KN09} that $(\log_23-\frac{4}{3},\frac{2}{3},\frac{2}{3})$ is not contained in $\mathcal{C}(p^B_{Y_1,Y_2|X})$. Note that $(H_b(\alpha),\frac{2}{3},\frac{2}{3})$ converges to $(\log_23-\frac{4}{3},\frac{2}{3},\frac{2}{3})$ as $\alpha\rightarrow H^{-1}_b(\log_23-\frac{4}{3})$. Since $\mathcal{C}(p_{Y_1,Y_2|X})$ is closed, it follows that $p_{(S_1,S_2)}$ violates (\ref{eq:violate}) with $\kappa=1$ (and consequently the necessary condition in Theorem \ref{thm:losslessnew} with $\kappa=1$) when $\alpha$ is sufficiently close to $H^{-1}_b(\log_23-\frac{4}{3})$.

This example indicates that choosing $U=S_0$ in Theorem \ref{thm:losslessnew} is not always optimal.
In this sense, the common part between $S_1$ and $S_2$ does not play a fundamental role in the new outer bound; see \cite{MK09} for a related observation.

\section{Approach II: Comparison of Capacity Regions}\label{sec:II}


\subsection{The Lossy Source Broadcast Problem}

In a certain sense, the notion of dominance in Section \ref{sec:I} hinges upon the converse results for broadcast channels.  In this section we shall develop a different notion of dominance that is mainly based on the achievability results for broadcast channels. This notion is captured by the following lemma, which shows that every realizable virtual broadcast channel $p_{\hat{S}_1,\hat{S}_2|S}$  is dominated by the physical broadcast channel $p_{Y_1,Y_2|X}$  in the sense that Marton's inner bound of $p_{\hat{S}_1,\hat{S}_2|S}$ with input distribution $p_S$ is contained in the capacity region of $p_{Y_1,Y_2|X}$. The proof of this lemma is relegated to Appendix \ref{app:capacitycomparison}.

\begin{lemma}\label{lem:capacitycomparison}
If a virtual broadcast channel $p_{\hat{S}_1,\hat{S}_2|S}$ is realizable through the physical broadcast channel $p_{Y_1,Y_2|X}$ with bandwidth expansion ratio $\rho$, then
\begin{align*}
\mathcal{C}_{\textsf{in}}(p_S, p_{\hat{S}_1,\hat{S}_2|S})\subseteq\rho\mathcal{C}(p_{Y_1,Y_2|X}).
\end{align*}
\end{lemma}

The following necessary condition for the lossy source broadcast problem is a simple consequence of Lemma \ref{lem:capacitycomparison}.

\begin{theorem}\label{thm:capacitycomp}
If distortion pair $(d_1,d_2)$ is achievable under distortion measures $w_1$ and $w_2$ subject to bandwidth expansion constraint $\kappa$, then there exists a virtual broadcast channel $p_{\hat{S}_1,\hat{S}_2|S}$ with $\mathbb{E}[w_i(S,\hat{S}_i)]\leq d_i$, $i=1,2$, such that
\begin{align*}
\mathcal{C}_{\textsf{in}}(p_S, p_{\hat{S}_1,\hat{S}_2|S})\subseteq\kappa\mathcal{C}(p_{Y_1,Y_2|X}).
\end{align*}
\end{theorem}
\begin{IEEEproof}
Let $(d_1,d_2)$ be a distortion pair that is achievable under distortion measures $w_1$ and $w_2$ subject to bandwidth expansion constraint $\kappa$. In view of Definition \ref{def:systemA2} and the discussion in Section \ref{sec:versus}, for every $\epsilon>0$, there exists a virtual broadcast channel $p_{\hat{S}^{(\epsilon)}_1,\hat{S}^{(\epsilon)}_2|S}$ realizable through the physical broadcast channel $p_{Y_1,Y_2|X}$ with bandwidth expansion ratio $\rho\leq\kappa+\epsilon$ such that $\mathbb{E}[w_i(S,\hat{S}^{(\epsilon)}_i)]\leq d_i+\epsilon$, $i=1,2$.
It follows from Lemma \ref{lem:capacitycomparison} that, for such $p_{\hat{S}^{(\epsilon)}_1,\hat{S}^{(\epsilon)}_2|S}$, we have
\begin{align*}
\mathcal{C}_{\textsf{in}}(p_S, p_{\hat{S}^{(\epsilon)}_1,\hat{S}^{(\epsilon)}_2|S})\subseteq(\kappa+\epsilon)\mathcal{C}(p_{Y_1,Y_2|X}).
\end{align*}
Since  $\{p_{\hat{S}^{(\epsilon)}_1,\hat{S}^{(\epsilon)}_2|S}:\epsilon>0\}$ can be viewed as a subset of $\{\pi\in\mathbb{R}^{|\mathcal{S}|\times|\hat{\mathcal{S}}_1|\times|\hat{\mathcal{S}}_2|}_+:\sum_{\hat{s}_1\in\hat{\mathcal{S}}_1, \hat{s}_2\in\hat{\mathcal{S}}_2}\pi(s,\hat{s}_1,\hat{s}_2)=1, s\in\mathcal{S}\}$, which is compact under the Euclidean distance, one can find a sequence $\epsilon_1, \epsilon_2, \cdots$ converging to zero such that
\begin{align*}
\lim\limits_{k\rightarrow\infty}p_{\hat{S}^{(\epsilon_k)}_1,\hat{S}^{(\epsilon_k)}_2|S}= p_{\hat{S}_1,\hat{S}_2|S}
\end{align*}
for some $p_{\hat{S}_1,\hat{S}_2|S}$ with $\mathbb{E}[w_i(S,\hat{S}_i)]\leq d_i$, $i=1,2$. Now the proof can be completed via a simple limiting argument.
\end{IEEEproof}




\subsection{Application to the Problem of Broadcasting Correlated Sources}

In view of (\ref{eq:inview}), one can readily deduce from Theorem \ref{thm:capacitycomp} the following outer bound on the admissible source region.

\begin{theorem}\label{thm:losslesscomp}
If $p_{(S_1,S_2)}$ is admissible for broadcast channel $p_{Y_1,Y_2|X}$ subject to bandwidth expansion constraint $\kappa$, then
\begin{align}
\mathcal{C}(p_{(S_1,S_2)})\subseteq\kappa\mathcal{C}(p_{Y_1,Y_2|X}).\label{eq:implyy2}
\end{align}
\end{theorem}

The following result provides a complete characterization of the source admissible region and a rigorous justification of the source-message correspondence in the case where $S_1\leftrightarrow S_0\leftrightarrow S_2$ form a Markov chain.
\begin{corollary}\label{cor:Markov}
A source distribution $p_{(S_1,S_2)}$ with $S_1\leftrightarrow S_0\leftrightarrow S_2$ forming a Markov chain is admissible for  broadcast channel $p_{Y_1,Y_2|X}$ subject to bandwidth expansion constraint $\kappa$ if and only if
\begin{align}
(H(S_0),H(S_1|S_0),H(S_2|S_0))\in\kappa\mathcal{C}(p_{Y_1,Y_2|X}).\label{eq:comparee2}
\end{align}
\end{corollary}
\begin{proof}
The proof of the ``if" part is based on a simple separation-based scheme. The transmitter first compresses $S^m_0$ via entropy coding and maps the resulting bits to the common message $M_0$; given $S^m_0$, the transmitter further compresses $S^m_i$ via conditional entropy coding and maps the resulting bits to the private message $M_i$, $i=1,2$. Note that (\ref{eq:comparee2}) ensures the existence of a good broadcast channel code such that receiver $i$ can recover $(M_0,M_i)$ and consequently $S^m_i$ with high probability, $i=1,2$.

The ``only if" part follows from Theorem \ref{thm:losslesscomp} as well as the fact that $(H(S_0),H(S_1|S_0),H(S_2|S_0))\in\mathcal{C}(p_{(S_1,S_2)})$ when $S_1\leftrightarrow S_0\leftrightarrow S_2$ form a Markov chain.
\end{proof}

In view of Theorem \ref{thm:degenerate} and Corollary \ref{cor:Markov}, the necessary conditions in Theorem \ref{thm:losslessnew} and Theorem \ref{thm:losslesscomp} are equivalent to (\ref{eq:degenerate}) and (\ref{eq:comparee2}), respectively, when $S_1\leftrightarrow S_0\leftrightarrow S_2$ form a Markov chain. It is known \cite{GGNY11} that in general $\mathcal{C}_{\textsf{out}}(p_{Y_1,Y_2|X})$ can be strictly larger than $\mathcal{C}(p_{Y_1,Y_2|X})$. So it is possible to find an example for which
\begin{align*}
&(H(S_0),H(S_1|S_0),H(S_2|S_0))\notin\kappa\mathcal{C}(p_{Y_1,Y_2|X}),\\
&(H(S_0),H(S_1|S_0),H(S_2|S_0))\in\kappa\mathcal{C}_{\textsf{out}}(p_{Y_1,Y_2|X}).
\end{align*}
This means\footnote{We believe that Theorem \ref{thm:losslessnew} also cannot be deduced from Theorem \ref{thm:losslesscomp}.} that Theorem \ref{thm:losslesscomp} cannot be deduced from Theorem \ref{thm:losslessnew}.

Note that both (\ref{eq:implyy1}) and (\ref{eq:implyy2}) imply
\begin{align}
\mathcal{C}(p_{(S_1,S_2)})\subseteq\kappa\mathcal{C}_{\textsf{out}}(p_{Y_1,Y_2|X}).\label{eq:implyy}
\end{align}
We shall show that (\ref{eq:implyy}) suffices to recover several existing results. Let  $\mathcal{C}_{\textsf{D}}(p_{Y_1,Y_2|X})$ denote the capacity region of broadcast channel $p_{Y_1,Y_2|X}$ with degraded message sets, i.e.,
\begin{align*}
\mathcal{C}_{\textsf{D}}(p_{Y_1,Y_2|X})=\{(R_0,R_2): (R_0,0,R_2)\in\mathcal{C}(p_{Y_1,Y_2|X})\}.
\end{align*}
It is known \cite{KM77b} that $\mathcal{C}_{\textsf{D}}(p_{Y_1,Y_2|X})$ is given by the set of $(R_0,R_2)\in\mathbb{R}^{2}_+$ satisfying
\begin{align*}
&R_0\leq I(V;Y_1),\\
&R_2\leq I(X;Y_2|V),\\
&R_0+R_2\leq I(X;Y_2)
\end{align*}
for some $p_{V,X,Y_1,Y_2}=p_{V|X}p_Xp_{Y_1,Y_2|X}$ with $|\mathcal{V}|\leq|\mathcal{X}|+1$. Moreover, it can be verified that
\begin{align}
\{(R_0,R_2): (R_0,0,R_2)\in\mathcal{C}_{\textsf{out}}(p_{Y_1,Y_2|X})\}=\mathcal{C}_{\textsf{D}}(p_{Y_1,Y_2|X}).\label{eq:degradedouter}
\end{align}

 The following result is a special case of \cite[Theorems 2 and 3]{KK08}.

\begin{corollary}\label{cor:degradedsource}
A source distribution $p_{(S_1,S_2)}$ with $S_1$ being a deterministic function of $S_2$ is admissible for  broadcast channel $p_{Y_1,Y_2|X}$ subject to bandwidth expansion constraint $\kappa$ if and only if
\begin{align}
(H(S_1),H(S_2|S_1))\in\kappa\mathcal{C}_{\textsf{D}}(p_{Y_1,Y_2|X}).\label{eq:degs}
\end{align}
\end{corollary}
\begin{proof}
The proof of the ``if" part is based on a simple separation-based scheme. The transmitter first compresses $S^m_1$ via entropy coding and maps the resulting bits to the common message $M_0$; given $S^m_1$, the transmitter further compresses $S^m_2$ via conditional entropy coding and maps the resulting bits to the private message $M_2$. Note that (\ref{eq:degs}) ensures the existence of a good broadcast channel code such that receiver 1 can recover $M_0$ and consequently $S^m_1$ with high probability while receiver 2 can recover $(M_0,M_2)$ and consequently $S^m_2$ with high probability.

The ``only if" part follows by (\ref{eq:implyy}) and (\ref{eq:degradedouter}) as well as the fact that  $(H(S_1),0,H(S_2|S_1))\in\mathcal{C}(p_{(S_1,S_2)})$
when $S_1$ is a deterministic function of $S_2$.
\end{proof}

We say $p_{Y_2|X}$ is more capable than $p_{Y_1|X}$ if $I(X;Y_2)\geq I(X;Y_1)$ for all $p_X$ \cite{KM77} \cite[p. 121]{EGK11}.  It can be verified that
\begin{align}
\{(0,R_1,R_2)\in\mathcal{C}_{\textsf{out}}(p_{Y_1,Y_2|X})\}=\mathcal{C}_{\textsf{D}}(p_{Y_1,Y_2|X}) \label{eq:verifyy2}
\end{align}
when $p_{Y_2|X}$ is more capable than $p_{Y_1|X}$.


The following result \cite[Theorem 4]{KLS09} is a dual version of Corollary \ref{cor:degradedsource}.

\begin{corollary}\label{cor:morecapable}
A source distribution $p_{(S_1,S_2)}$ is admissible for  broadcast channel $p_{Y_1,Y_2|X}$ (with $p_{Y_2|X}$ more capable than $p_{Y_1|X}$) subject to bandwidth expansion constraint $\kappa$ if and only if (\ref{eq:degs}) holds.
\end{corollary}
\begin{proof}
The proof of the ``if" part is the same as that for Corollary \ref{cor:degradedsource}.
The ``only if" part follows by (\ref{eq:implyy}) and (\ref{eq:verifyy2}) as well as the fact that $(0,H(S_1),H(S_2|S_1))\in\mathcal{C}(p_{(S_1,S_2)})$.
\end{proof}

\section{Conclusion}\label{sec:conclusion}

We have established two necessary conditions for the lossy source broadcast problem (Theorem \ref{thm:ordering} and Theorem \ref{thm:capacitycomp}), from which new outer bounds on the admissible source region (Theorem \ref{thm:losslessnew} and Theorem \ref{thm:losslesscomp}) are deduced. It is expected that the idea of deriving converse results via suitable comparisons between the virtual channel (induced by the source(s) and the reconstruction(s)) and the physical channel has potential applications beyond the lossy source broadcast problem considered in the present paper.




\appendices

\section{Proof of Lemma \ref{lem:ordering}}\label{app:ordering}

Let $(X^{\rho m}, Y^{\rho m}_1, Y^{\rho m}_2, \hat{S}^m_1, \hat{S}^m_2)$ be jointly distributed with $S^m$ according to
\begin{align}
p_{S^m}(s^m)p_{X^{\rho m}|S^m}(x^{\rho m}|s^m)p_{Y^{\rho m}_1,Y^{\rho m}_2|X^{\rho m}}(y^{\rho m}_1,y^{\rho m}_2|x^{\rho m})\prod\limits_{i=1}^2p_{\hat{S}^m_i|Y^{\rho m}_i}(\hat{s}^m_i|y^{\rho m}_i),\label{eq:multijoint}
\end{align}
where
\begin{align*}
&p_{S^m}(s^m)=\prod\limits_{t=1}^mp_S(s(t)),\\
&p_{Y^{\rho m}_1,Y^{\rho m}_2|X^{\rho m}}(y^{\rho m}_1,y^{\rho m}_2|x^{\rho m})=\prod\limits_{q=1}^{\rho m}(y_1(q),y_2(q)|x(q)).
\end{align*}
Note that the induced conditional distribution\footnote{If $p_S(s)>0$ for all $s\in\mathcal{S}$, then $p_{\hat{S}^m_1,\hat{S}^m_2|S^m}$ is uniquely given by (\ref{eq:multicond}). If $p_S(s)=0$ for some $s\in\mathcal{S}$, then the conditional distribution in (\ref{eq:multicond}) is not the only one that is compatible with the joint distribution in (\ref{eq:multijoint}); in this case we simply use (\ref{eq:multicond}) as the definition of $p_{\hat{S}^m_1,\hat{S}^m_2|S^m}$.} $p_{\hat{S}^m_1,\hat{S}^m_2|S^m}$ is degraded with respect to $p_{Y^{\rho m}_1,Y^{\rho m}_2|X^{\rho m}}$; in fact, every $p_{\hat{S}^m_1,\hat{S}^m_2|S^m}$ degraded with respect to $p_{Y^{\rho m}_1,Y^{\rho m}_2|X^{\rho m}}$ can be obtained in this way.

Let $(U^m_{0},\cdots,U^m_{L})$ be jointly distributed with $(S^m, X^{\rho m}, Y^{\rho m}_1, Y^{\rho m}_2, \hat{S}^m_{1},\hat{S}^m_{2})$ such that $(U^m_{0},\cdots,U^m_{L})\leftrightarrow S^m\leftrightarrow(X^{\rho m}, Y^{\rho m}_1, Y^{\rho m}_2,\hat{S}^m_{1},\hat{S}^m_{2})$ form a Markov chain,
and $(U_0(t),\cdots,U_L(t),S(t))$, $t=1,\cdots,m$, are independent and identically distributed.
Let $T$ be a random variable independent of $(U^m_{0,1},\cdots,U^m_{L},S^m,\hat{S}^m_{1},\hat{S}^m_{2})$ and uniformly distributed over $\{1,\cdots,m\}$. Define
\begin{align*}
&U_i=U_i(T),\quad i=0,\cdots,L,\\
&S=S(T),\\
&\hat{S}_i=\hat{S}_i(T),\quad i=1,2.
\end{align*}
The following properties of $(U_0,\cdots,U_L,S,\hat{S}_1,\hat{S}_2)$ can be easily verified:
\begin{enumerate}
\item the distribution of $(U_0,\cdots,U_L,S)$ is identical with that of $(U_0(t),U_1(t),U_2(t),S(t))$ for every $t$;

\item $(U_0,\cdots,U_L)\leftrightarrow S\leftrightarrow (\hat{S}_1,\hat{S}_2)$ form a Markov chain;

\item $p_{\hat{S}_1,\hat{S}_2|S}$ is the single-letterized version\footnote{Strictly speaking, $p_{\hat{S}_1,\hat{S}_2|S}(\cdot,\cdot|s)$ is uniquely specified only for those $s\in\{s'\in\mathcal{S}:p_S(s')>0\}$. However, this suffices for our purpose since the results in the present paper depend on $p_{\hat{S}_1,\hat{S}_2|S}$ only through $p_Sp_{\hat{S}_1,\hat{S}_2|S}$.} of $p_{\hat{S}^m_1,\hat{S}^m_2|S^m}$.
\end{enumerate}
Note that
\begin{align*}
&\sum\limits_{i=1}^kI(U^m_{\mathcal{A}_i};\hat{S}^m_{a(i)}|U^m_{\cup_{j=1}^{i-1}\mathcal{A}_j})\\
&=\sum\limits_{i=1}^k\sum\limits_{t=1}^mI(U_{\mathcal{A}_i}(t);\hat{S}^m_{a(i)}|U^m_{\cup_{j=1}^{i-1}\mathcal{A}_j},U^{t-1}_{\mathcal{A}_i})\\
&=\sum\limits_{i=1}^k\sum\limits_{t=1}^mI(U_{\mathcal{A}_i}(t);\hat{S}^m_{a(i)},U^{t-1}_{\cup_{j=1}^{i-1}\mathcal{A}_j},U^{m}_{\cup_{j=1}^{i-1}\mathcal{A}_j,t+1},U^{t-1}_{\mathcal{A}_i}|U_{\cup_{j=1}^{i-1}\mathcal{A}_j}(t))\\
&\geq\sum\limits_{i=1}^k\sum\limits_{t=1}^mI(U_{\mathcal{A}_i}(t);\hat{S}_{a(i)}(t)|U_{\cup_{j=1}^{i-1}\mathcal{A}_j}(t))\\
&=m\sum\limits_{i=1}^k I(U_{\mathcal{A}_i}(T);\hat{S}_{a(i)}(T)|U_{\cup_{j=1}^{i-1}\mathcal{A}_j}(T),T)\\
&=m\sum\limits_{i=1}^k I(U_{\mathcal{A}_i}(T);\hat{S}_{a(i)}(T),T|U_{\cup_{j=1}^{i-1}\mathcal{A}_j}(T))\\
&\geq m\sum\limits_{i=1}^kI(U_{\mathcal{A}_i}(T);\hat{S}_{a(i)}(T)|U_{\cup_{j=1}^{i-1}\mathcal{A}_j}(T))\\
&=m\sum\limits_{i=1}^k I(U_{\mathcal{A}_i};\hat{S}_{a(i)}|U_{\cup_{j=1}^{i-1}\mathcal{A}_j}).
\end{align*}
On the other hand, we have
\begin{align}
\sum\limits_{i=1}^kI(U^m_{\mathcal{A}_i};\hat{S}^m_{a(i)}|U^m_{\cup_{j=1}^{i-1}\mathcal{A}_j})\leq\sum\limits_{i=1}^kI(U^m_{\mathcal{A}_i};Y^{\rho m}_{a(i)}|U^m_{\cup_{j=1}^{i-1}\mathcal{A}_j}).\label{eq:ontheother}
\end{align}
We shall show that, for $l=1,\cdots,k$,
\begin{align}
&\sum\limits_{i=l}^kI(U^m_{\mathcal{A}_i};Y^{\rho m}_{a(i)}|U^m_{\cup_{j=1}^{i-1}\mathcal{A}_j})\nonumber\\
&\leq \sum\limits_{q=1}^{\rho m}I(Y^{\rho m}_{2,q+1};Y_1(q)|U^m_{\cup_{j=1}^{l-1}\mathcal{A}_j},Y^{q-1}_{1})+\sum\limits_{i=l}^k\sum\limits_{q=1}^{\rho m}I(U^m_{\mathcal{A}_i};Y_{a(i)}(q)|U^m_{\cup_{j=1}^{i-1}\mathcal{A}_j},Y^{q-1}_{1},Y^{\rho m}_{2,q+1}),\label{eq:equiv1}
\end{align}
which, in light of Csisz\'{a}r sum identity \cite[p. 25]{EGK11}, is equivalent to
\begin{align}
&\sum\limits_{i=l}^kI(U^m_{\mathcal{A}_i};Y^{\rho m}_{a(i)}|U^m_{\cup_{j=1}^{i-1}\mathcal{A}_j})\nonumber\\
&\leq\sum\limits_{q=1}^{\rho m}I(Y^{q-1}_{1};Y_2(q)|U^m_{\cup_{j=1}^{l-1}\mathcal{A}_j},Y^{\rho m}_{2,q+1})+\sum\limits_{i=l}^k\sum\limits_{q=1}^{\rho m}I(U^m_{\mathcal{A}_i};Y_{a(i)}(q)|U^m_{\cup_{j=1}^{i-1}\mathcal{A}_j},Y^{q-1}_{1},Y^{\rho m}_{2,q+1}).\label{eq:equiv2}
\end{align}
First consider the case $l=k$. If $a(k)=1$, we have
\begin{align*}
&I(U^m_{\mathcal{A}_i};Y^{\rho m}_{a(k)}|U^m_{\cup_{j=1}^{k-1}\mathcal{A}_j})\\
&=\sum\limits_{q=1}^{\rho m}I(U^m_{\mathcal{A}_i};Y_1(q)|U^m_{\cup_{j=1}^{k-1}\mathcal{A}_j},Y^{q-1}_{1})\\
&\leq\sum\limits_{q=1}^{\rho m}I(U^m_{\mathcal{A}_i},Y^{\rho m}_{2,q+1};Y_1(q)|U^m_{\cup_{j=1}^{k-1}\mathcal{A}_j},Y^{q-1}_{1})\\
&=\sum\limits_{q=1}^{\rho m}I(Y^{\rho m}_{2,q+1};Y_1(q)|U^m_{\cup_{j=1}^{k-1}\mathcal{A}_j},Y^{q-1}_{1})+\sum\limits_{q=1}^{\rho m}I(U^m_{\mathcal{A}_i};Y_1(q)|U^m_{\cup_{j=1}^{k-1}\mathcal{A}_j},Y^{q-1}_{1},Y^{\rho m}_{2,q+1});
\end{align*}
if $a(k)=2$, we have
\begin{align*}
&I(U^m_{\mathcal{A}_i};Y^{\rho m}_{a(k)}|U^m_{\cup_{j=1}^{k-1}\mathcal{A}_j})\\
&=\sum\limits_{q=1}^{\rho m}I(U^m_{\mathcal{A}_i};Y_2(q)|U^m_{\cup_{j=1}^{k-1}\mathcal{A}_j},Y^{\rho m}_{2,q+1})\\
&\leq\sum\limits_{q=1}^{\rho m}I(U^m_{\mathcal{A}_i},Y^{q-1}_{1};Y_2(q)|U^m_{\cup_{j=1}^{k-1}\mathcal{A}_j},Y^{\rho m}_{2,q+1})\\
&=\sum\limits_{q=1}^{\rho m}I(Y^{q-1}_{1};Y_2(q)|U^m_{\cup_{j=1}^{k-1}\mathcal{A}_j},Y^{\rho m}_{2,q+1})+\sum\limits_{q=1}^{\rho m}I(U^m_{\mathcal{A}_i};Y_2(q)|U^m_{\cup_{j=1}^{k-1}\mathcal{A}_j},Y^{q-1}_{1},Y^{\rho m}_{2,q+1}).
\end{align*}
Therefore, (\ref{eq:equiv1}) and (\ref{eq:equiv2}) hold when $l=k$. Now we proceed by induction on $l$. If $a(l)=1$, we have
\begin{align}
&\sum\limits_{i=l}^kI(U^m_{\mathcal{A}_i};Y^{\rho m}_{a(i)}|U^m_{\cup_{j=1}^{i-1}\mathcal{A}_j})\nonumber\\
&=I(U^m_{\mathcal{A}_l};Y^{\rho m}_{1}|U^m_{\cup_{j=1}^{l-1}\mathcal{A}_j})+\sum\limits_{i=l+1}^kI(U^m_{\mathcal{A}_i};Y^{\rho m}_{a(i)}|U^m_{\cup_{j=1}^{i-1}\mathcal{A}_j})\nonumber\\
&\leq I(U^m_{\mathcal{A}_l};Y^{\rho m}_{1}|U^m_{\cup_{j=1}^{l-1}\mathcal{A}_j})+\sum\limits_{q=1}^{\rho m}I(Y^{\rho m}_{2,q+1};Y_1(q)|U^m_{\cup_{j=1}^{l}\mathcal{A}_j},Y^{q-1}_{1})\nonumber\\
&\quad+\sum\limits_{i=l+1}^k\sum\limits_{q=1}^{\rho m}I(U^m_{\mathcal{A}_i};Y_{a(i)}(q)|U^m_{\cup_{j=1}^{i-1}\mathcal{A}_j},Y^{q-1}_{1},Y^{\rho m}_{2,q+1})\label{eq:inductionhyp}\\
&=\sum\limits_{q=1}^{\rho m}I(U^m_{\mathcal{A}_l};Y_{1}(q)|U^m_{\cup_{j=1}^{l-1}\mathcal{A}_j},Y^{q-1}_{1})+\sum\limits_{q=1}^{\rho m}I(Y^{\rho m}_{2,q+1};Y_1(q)|U^m_{\cup_{j=1}^{l}\mathcal{A}_j},Y^{q-1}_{1})\nonumber\\
&\quad+\sum\limits_{i=l+1}^k\sum\limits_{q=1}^{\rho m}I(U^m_{\mathcal{A}_i};Y_{a(i)}(q)|U^m_{\cup_{j=1}^{i-1}\mathcal{A}_j},Y^{q-1}_{1},Y^{\rho m}_{2,q+1})\nonumber\\
&=\sum\limits_{q=1}^{\rho m}I(U^m_{\mathcal{A}_l},Y^{\rho m}_{2,q+1};Y_{1}(q)|U^m_{\cup_{j=1}^{l-1}\mathcal{A}_j},Y^{q-1}_{1})+\sum\limits_{i=l+1}^k\sum\limits_{q=1}^{\rho m}I(U^m_{\mathcal{A}_i};Y_{a(i)}(q)|U^m_{\cup_{j=1}^{i-1}\mathcal{A}_j},Y^{q-1}_{1},Y^{\rho m}_{2,q+1})\nonumber\\
&=\sum\limits_{q=1}^{\rho m}I(Y^{\rho m}_{2,q+1};Y_1(q)|U^m_{\cup_{j=1}^{l-1}\mathcal{A}_j},Y^{q-1}_{1})+\sum\limits_{i=l}^k\sum\limits_{q=1}^{\rho m}I(U^m_{\mathcal{A}_i};Y_{a(i)}(q)|U^m_{\cup_{j=1}^{i-1}\mathcal{A}_j},Y^{q-1}_{1},Y^{\rho m}_{2,q+1}),\nonumber
\end{align}
where (\ref{eq:inductionhyp}) follows by the induction hypothesis. Therefore, (\ref{eq:equiv1}) holds when $a(l)=1$. Similarly, it can be shown that (\ref{eq:equiv2}) holds when $a(l)=2$. This finishes the induction argument.

Let $Q$ be a random variable independent of $(U^m_{0,1},\cdots,U^m_{L},X^{\rho m},Y^{\rho m}_{1},Y^{\rho m}_{2})$ and uniformly distributed over $\{1,\cdots,\rho m\}$. Define
\begin{align*}
&V_i=(U^m_{i},Y^{Q-1}_{1},Y^{\rho m}_{2,Q+1},Q),\quad i=0,\cdots,L,\\
&X=X(Q),\\
&Y_i=Y_i(Q),\quad i=1,2.
\end{align*}
It is clear that $(V_0,\cdots,V_L)\leftrightarrow X\leftrightarrow(Y_1,Y_2)$ form a Markov chain; moreover, $p_X$ does not depend on the choice of $p_{U_0,\cdots,U_L|S}$.
Continuing from (\ref{eq:ontheother}),
\begin{align}
&\sum\limits_{i=1}^kI(U^m_{\mathcal{A}_i};\hat{S}^m_{a(i)}|U^m_{\cup_{j=1}^{i-1}\mathcal{A}_j})\nonumber\\
&\leq\sum\limits_{q=1}^{\rho m}I(U^m_{\mathcal{A}_1},Y^{q-1}_{1},Y^{\rho m}_{2,q+1};Y_{a(1)}(q))+\sum\limits_{i=2}^k\sum\limits_{q=1}^{\rho m}I(U^m_{\mathcal{A}_i};Y_{a(i)}(q)|U^m_{\cup_{j=1}^{i-1}\mathcal{A}_j},Y^{q-1}_{1},Y^{\rho m}_{2,q+1})\label{eq:byind}\\
&=\rho mI(U^m_{\mathcal{A}_1},Y^{Q-1}_{1},Y^{\rho m}_{2,Q+1};Y_{a(1)}(Q)|Q)+\rho m\sum\limits_{i=2}^kI(U^m_{\mathcal{A}_i};Y_{a(i)}(Q)|U^m_{\cup_{j=1}^{i-1}\mathcal{A}_j},Y^{Q-1}_{1},Y^{\rho m}_{2,Q+1},Q)\nonumber\\
&\leq \rho mI(U^m_{\mathcal{A}_1},Y^{Q-1}_{1},Y^{\rho m}_{2,Q+1},Q;Y_{a(1)}(Q))+\rho m\sum\limits_{i=2}^kI(U^m_{\mathcal{A}_i};Y_{a(i)}(Q)|U^m_{\cup_{j=1}^{i-1}\mathcal{A}_j},Y^{Q-1}_{1},Y^{\rho m}_{2,Q+1},Q)\nonumber\\
&=\rho m\sum\limits_{i=1}^kI(V_{\mathcal{A}_i};Y_{a(i)}|V_{\cup_{j=1}^{i-1}\mathcal{A}_j}),\nonumber
\end{align}
where (\ref{eq:byind}) is due to (\ref{eq:equiv1}) and (\ref{eq:equiv2}) as well as the fact that
\begin{align*}
&\sum\limits_{q=1}^{\rho m}I(Y^{\rho m}_{2,q+1};Y_1(q)|Y^{q-1}_{1})+\sum\limits_{q=1}^{\rho m}I(U^m_{\mathcal{A}_1};Y_{a(1)}(q)|Y^{q-1}_{1},Y^{\rho m}_{2,q+1})\\
&=\sum\limits_{q=1}^{\rho m}I(Y^{q-1}_{1};Y_2(q)|Y^{\rho m}_{2,q+1})+\sum\limits_{q=1}^{\rho m}I(U^m_{\mathcal{A}_1};Y_{a(1)}(q)|Y^{q-1}_{1},Y^{\rho m}_{2,q+1})\\
&\leq\sum\limits_{q=1}^{\rho m}I(U^m_{\mathcal{A}_1},Y^{q-1}_{1},Y^{\rho m}_{2,q+1};Y_{a(1)}(q)).
\end{align*}
This completes the proof of Theorem \ref{lem:ordering}.

\section{Proof of Theorem \ref{thm:ordering}}\label{app:TheoremI}

According to  Lemma \ref{lem:ordering}, for every virtual broadcast channel $p_{\hat{S}_1,\hat{S}_2|S}$ realizable through the physical broadcast channel $p_{Y_1,Y_2|X}$ with bandwidth expansion ratio $\rho$, there exists an input distribution $p_X$ such that, for any $p_{U_0,\cdots,U_L,S,\hat{S}_1,\hat{S}_2}=p_{U_0,\cdots,U_L|S}p_Sp_{\hat{S}_1,\hat{S}_2|S}$, one can find $p_{V_0,\cdots,V_L,X,Y_1,Y_2}=p_{V_0,\cdots,V_L|X}p_Xp_{Y_1,Y_2|X}$ satisfying
\begin{align}
\sum\limits_{i=1}^kI(U_{\mathcal{A}_i};\hat{S}_{a(i)}|U_{\cup_{j=1}^{i-1}\mathcal{A}_j})\leq\rho\sum\limits_{i=1}^kI(V_{\mathcal{A}_i};Y_{a(i)}|V_{\cup_{j=1}^{i-1}\mathcal{A}_j})\label{eq:keyeq}
\end{align}
for any $\mathcal{A}_i\subseteq\{0,\cdots,L\}$ and $a(i)\in\{1,2\}$, $i=1,\cdots,k$. Now choose $L=2$. Setting $k=1$, $\mathcal{A}_1=\{0\}$, and $a(1)=1$ in (\ref{eq:keyeq}) gives
\begin{align}
I(U_0;\hat{S}_1)\leq\rho I(V_0;Y_1).\label{eq:constr1}
\end{align}
Setting $k=1$, $\mathcal{A}_1=\{0\}$, and $a(1)=2$ in (\ref{eq:keyeq}) gives
\begin{align}
I(U_0;\hat{S}_2)\leq\rho I(V_0;Y_2).\label{eq:constr2}
\end{align}
Setting $k=1$, $\mathcal{A}_1=\{0,1\}$  and $a(1)=1$ in (\ref{eq:keyeq}) gives
\begin{align}
I(U_0,U_1;\hat{S}_1)\leq\rho I(V_0,V_1;Y_1).\label{eq:constr3}
\end{align}
Setting $k=1$, $\mathcal{A}_1=\{0,2\}$, and $a(1)=2$ in (\ref{eq:keyeq}) gives
\begin{align}
I(U_0,U_2;\hat{S}_2)\leq\rho I(V_0,V_2;Y_2).\label{eq:constr4}
\end{align}
Setting $k=2$, $\mathcal{A}_1=\{0\}$, $\mathcal{A}_2=\{2\}$, $a(1)=1$, and $a(2)=2$ in (\ref{eq:keyeq}) gives
\begin{align}
I(U_0;\hat{S}_1)+I(U_2;\hat{S}_2|U_0)\leq\rho[I(V_0;Y_1)+I(V_2;Y_2|V_0)].\label{eq:constr5}
\end{align}
Setting $k=2$, $\mathcal{A}_1=\{0\}$, $\mathcal{A}_2=\{1\}$, $a(1)=2$, and $a(2)=1$ in (\ref{eq:keyeq}) gives
\begin{align}
I(U_0;\hat{S}_2)+I(U_1;\hat{S}_1|U_0)\leq\rho[I(V_0;Y_2)+I(V_1;Y_1|V_0)].\label{eq:constr6}
\end{align}
Setting $k=2$, $\mathcal{A}_1=\{0,1\}$, $\mathcal{A}_2=\{2\}$, $a(1)=1$, and $a(2)=2$ in (\ref{eq:keyeq}) gives
\begin{align}
I(U_0,U_1;\hat{S}_1)+I(U_2;\hat{S}_2|U_0,U_1)\leq\rho[I(V_0,V_1;Y_1)+I(V_2;Y_2|V_0,V_1)].\label{eq:constr7}
\end{align}
Setting $k=2$, $\mathcal{A}_1=\{0,2\}$, $\mathcal{A}_2=\{1\}$, $a(1)=2$, and $a(2)=1$ in (\ref{eq:keyeq}) gives
\begin{align}
I(U_0,U_2;\hat{S}_2)+I(U_1;\hat{S}_1|U_0,U_2)\leq\rho[I(V_0,V_2;Y_2)+I(V_1;Y_1|V_0,V_2)].\label{eq:constr8}
\end{align}
Setting $k=3$, $\mathcal{A}_1=\{0\}$, $\mathcal{A}_2=\{2\}$, $\mathcal{A}_3=\{1\}$, $a(1)=a(3)=1$, and $a(2)=2$ in (\ref{eq:keyeq}) gives
\begin{align}
I(U_0;\hat{S}_1)+I(U_2;\hat{S}_2|U_0)+I(U_1;\hat{S}_1|U_0,U_2)\leq\rho[I(V_0;Y_1)+I(V_2;Y_2|V_0)+I(V_1;Y_1|V_0,V_2)].\label{eq:constr9}
\end{align}
Setting $k=3$, $\mathcal{A}_1=\{0\}$, $\mathcal{A}_2=\{1\}$, $\mathcal{A}_3=\{2\}$, $a(1)=a(3)=2$, and $a(2)=1$ in (\ref{eq:keyeq}) gives
\begin{align}
I(U_0;\hat{S}_2)+I(U_1;\hat{S}_1|U_0)+I(U_2;\hat{S}_2|U_0,U_1)\leq\rho[I(V_0;Y_2)+I(V_1;Y_1|V_0)+I(V_2;Y_2|V_0,V_1)].\label{eq:constr10}
\end{align}
Let $\mathcal{R}(p_S,p_{\hat{S}_1,\hat{S}_2|S})$ denote the set of $(r_1,\cdots,r_{10})\in\mathbb{R}^{10}_+$ satisfying
\begin{align*}
&r_1\leq I(U_0;\hat{S}_1),\\
&r_2\leq I(U_0;\hat{S}_2),\\
&r_3\leq I(U_0,U_1;\hat{S}_1),\\
&r_4\leq I(U_0,U_2;\hat{S}_2),\\
&r_5\leq I(U_0;\hat{S}_1)+I(U_2;\hat{S}_2|U_0),\\
&r_6\leq I(U_0;\hat{S}_2)+I(U_1;\hat{S}_1|U_0),\\
&r_7\leq I(U_0,U_1;\hat{S}_1)+I(S;\hat{S}_2|U_0,U_1),\\
&r_8\leq I(U_0,U_2;\hat{S}_2)+I(S;\hat{S}_1|U_0,U_2),\\
&r_9\leq I(U_0;\hat{S}_1)+I(U_2;\hat{S}_2|U_0)+I(S;\hat{S}_1|U_0,U_2),\\
&r_{10}\leq I(U_0;\hat{S}_2)+I(U_1;\hat{S}_1|U_0)+I(S;\hat{S}_2|U_0,U_1)
\end{align*}
for some $p_{U_0,U_1,U_2,S,\hat{S}_1,\hat{S}_2}=p_{U_0,U_1,U_2|S}p_Sp_{\hat{S}_1,\hat{S}_2|S}$; analogously,
let $\mathcal{R}(p_X,p_{Y_1,Y_2|X})$ denote the set of $(r_1,\cdots,r_{10})\in\mathbb{R}^{10}_+$ satisfying
\begin{align*}
&r_1\leq I(V_0;Y_1),\\
&r_2\leq I(V_0;Y_2),\\
&r_3\leq I(V_0,V_1;Y_1),\\
&r_4\leq I(V_0,V_2;Y_2),\\
&r_5\leq I(V_0;Y_1)+I(V_2;Y_2|V_0),\\
&r_6\leq I(V_0;Y_2)+I(V_1;Y_1|V_0),\\
&r_7\leq I(V_0,V_1;Y_1)+I(X;Y_2|V_0,V_1),\\
&r_8\leq I(V_0,V_2;Y_2)+I(X;Y_1|V_0,V_2),\\
&r_9\leq I(V_0;Y_1)+I(V_2;Y_2|V_0)+I(X;Y_1|V_0,V_2),\\
&r_{10}\leq I(V_0;Y_2)+I(V_1;Y_1|V_0)+I(X;Y_2|V_0,V_1)
\end{align*}
for some $p_{V_0,V_1,V_2,X,Y_1,Y_2}=p_{V_0,V_1,V_2|X}p_Xp_{Y_1,Y_2|X}$.
It can be shown (see \cite[Remark 3.6]{NEG07}) that (\ref{eq:constr1})-(\ref{eq:constr10}) can be stated equivalently as
\begin{align}
\mathcal{R}(p_S,p_{\hat{S}_1,\hat{S}_2|S})\subseteq\rho\mathcal{R}(p_X,p_{Y_1,Y_2|X}).\label{eq:eqstate}
\end{align}
Moreover, the following argument by Nair \cite{Nair13p} indicates that, to compute $\mathcal{R}(p_S,p_{\hat{S}_1,\hat{S}_2|S})$ and $\mathcal{R}(p_X,p_{Y_1,Y_2|X})$, it suffices
to consider $|\mathcal{U}_0|\leq|\mathcal{S}|+5$, $|\mathcal{V}_0|\leq|\mathcal{X}|+5$, $|\mathcal{U}_i|\leq |\mathcal{S}|$, and $|\mathcal{V}_i|\leq |\mathcal{X}|$, $i=1,2$.
We shall only give the proof for $\mathcal{R}(p_S,p_{\hat{S}_1,\hat{S}_2|S})$ since $\mathcal{R}(p_X,p_{Y_1,Y_2|X})$ can be treated in the same way. The main idea is that it suffices for $U_1$ and $U_2$ to preserve the extreme points of $\mathcal{R}(p_S,p_{\hat{S}_1,\hat{S}_2|S})$, and then $U_0$ can be used to convexify the region. Note that every convex combination of the constraints in the definition of $\mathcal{R}(p_S,p_{\hat{S}_1,\hat{S}_2|S})$ can be written in the form
\begin{align*}
&\lambda_1H(\hat{S}_1)+\lambda_2H(\hat{S}_2)+\lambda_3I(S;\hat{S}_1)+\lambda_4I(S;\hat{S}_2)+\lambda_5H(\hat{S}_1|U_0)+\lambda_6H(\hat{S}_2|U_0)\\
&+\lambda_7H(\hat{S}_1|U_0,U_1)+\lambda_8H(\hat{S}_2|U_0,U_2)+\lambda_9H(\hat{S}_1|U_0,U_2)+\lambda_{10}H(\hat{S}_2|U_0,U_1),
\end{align*}
which depends on $p_{U_0,U_1,U_2,S}$ only through $p_{U_0,U_1,S}$ and $p_{U_0,U_2,S}$.
First fix $p_{U_0}$. For every $U_0=u_0$, one can find $p_{S,U_1|U_0}(\cdot,\cdot|u_0)$ with $|\mathcal{U}_1|\leq|\mathcal{S}|$ that preserves $p_{S|U_0}(\cdot|u_0)$ and $\lambda_7H(\hat{S}_1|U_0=u_0,U_1)+\lambda_{10}H(\hat{S}_2|U_0=u_0,U_1)$; similarly, one can find $p_{S,U_2|U_0}(\cdot,\cdot|u_0)$ with $|\mathcal{U}_2|\leq|\mathcal{S}|$ that preserves $p_{S|U_0}(\cdot|u_0)$ and $\lambda_8H(\hat{S}_2|U_0=u_0,U_2)+\lambda_{9}H(\hat{S}_1|U_0=u_0,U_2)$. We can get a consistent joint distribution $p_{U_0,U_1,U_2,S}$ by setting $p_{U_0,U_1,U_2,S}=p_{U_0,S}p_{U_1|U_0,S}p_{U_2|U_0,S}$. Finally, it suffices to have $|\mathcal{U}_0|\leq|\mathcal{S}|+5$ for preserving $p_S$, $H(\hat{S}_i|U_0)$, $H(\hat{S}_i|U_0,U_1)$, and $H(\hat{S}_i|U_0,U_2)$, $i=1,2$.

Let $(d_1,d_2)$ be a distortion pair that is achievable under distortion measures $w_1$ and $w_2$ subject to bandwidth expansion constraint $\kappa$. In view of Definition \ref{def:systemA2} and the discussion in Section \ref{sec:versus}, for every $\epsilon>0$, there exists a virtual broadcast channel $p_{\hat{S}^{(\epsilon)}_1,\hat{S}^{(\epsilon)}_2|S}$ realizable through the physical broadcast channel $p_{Y_1,Y_2|X}$ with bandwidth expansion ratio $\rho\leq\kappa+\epsilon$ such that $\mathbb{E}[w_i(S,\hat{S}^{(\epsilon)}_i)]\leq d_i+\epsilon$, $i=1,2$. It follows by (\ref{eq:eqstate}) that, for such $p_{\hat{S}^{(\epsilon)}_1,\hat{S}^{(\epsilon)}_2|S}$, we have
\begin{align*}
\mathcal{R}(p_S,p_{\hat{S}^{(\epsilon)}_1,\hat{S}^{(\epsilon)}_2|S})\subseteq(\kappa+\epsilon)\mathcal{R}(p_{X^{(\epsilon)}},p_{Y_1,Y_2|X})
\end{align*}
for some $p_{X^{(\epsilon)}}$. Since  $\{(p_{\hat{S}^{(\epsilon)}_1,\hat{S}^{(\epsilon)}_2|S}, p_{X^{(\epsilon)}}):\epsilon>0\}$ can be viewed as a subset of $\{(\pi,\pi')\in\mathbb{R}^{|\mathcal{S}|\times|\hat{\mathcal{S}}_1|\times|\hat{\mathcal{S}}_2|}_+\times\mathbb{R}^{|\mathcal{X}|}_+:\sum_{\hat{s}_1\in\hat{\mathcal{S}}_1, \hat{s}_2\in\hat{\mathcal{S}}_2}\pi(s,\hat{s}_1,\hat{s}_2)=1, s\in\mathcal{S}, \mbox{ and }\sum_{x\in\mathcal{X}}\pi'(x)=1\}$, which is compact under the Euclidean distance, one can find a sequence $\epsilon_1, \epsilon_2, \cdots$ converging to zero such that
\begin{align*}
&\lim\limits_{k\rightarrow\infty}p_{\hat{S}^{(\epsilon_k)}_1,\hat{S}^{(\epsilon_k)}_2|S}= p_{\hat{S}_1,\hat{S}_2|S},\\
&\lim\limits_{k\rightarrow\infty}p_{X^{(\epsilon_k)}}=p_{X}
\end{align*}
for some $p_{\hat{S}_1,\hat{S}_2|S}$ with $\mathbb{E}[w_i(S,\hat{S}_i)]\leq d_i$, $i=1,2$, and $p_X$. Now a simple limiting argument yields
\begin{align}
\mathcal{R}(p_S,p_{\hat{S}_1,\hat{S}_2|S})\subseteq\kappa\mathcal{R}(p_{X},p_{Y_1,Y_2|X}).\label{eq:covering}
\end{align}
Note that $\mathcal{R}(p_{X},p_{Y_1,Y_2|X})$ is a convex set. As a consequence, (\ref{eq:covering}) holds if and only if $\kappa\mathcal{R}(p_{X},p_{Y_1,Y_2|X})$ contains all extreme points of $\mathcal{R}(p_S,p_{\hat{S}_1,\hat{S}_2|S})$. To realize all such extreme points, it suffices to consider $|\mathcal{U}_0|\leq|\mathcal{S}|$. This completes the proof of Theorem \ref{thm:ordering}.

\section{Proof of Theorem \ref{thm:degenerate}}\label{app:degenerate}

We shall only prove that (\ref{eq:degenerate}) implies the necessary condition in Theorem \ref{thm:losslessnew} when $S_1\leftrightarrow S_0\leftrightarrow S_2$ form a Markov chain since the other direction is straightforward.

Note that the necessary condition in Theorem \ref{thm:losslessnew} can be written equivalently as
\begin{align}
\mathcal{R}(p_{(S_1,S_2)})\subseteq\kappa\mathcal{R}(p_X, p_{Y_1,Y_2|X})\label{eq:ratecomp}
\end{align}
for some $p_X$, where $\mathcal{R}(p_{(S_1,S_2)})$ is the set of $(r_1,\cdots,r_{10})\in\mathbb{R}^{10}_+$ satisfying
\begin{align*}
&r_1\leq I(U;S_1),\\
&r_2\leq I(U;S_2),\\
&r_3\leq H(S_1),\\
&r_4\leq H(S_2),\\
&r_5\leq I(U;S_1)+H(S_2|U),\\
&r_6\leq I(U;S_2)+H(S_1|U),\\
&r_7\leq I(U;S_1)+H(S_1,S_2|U),\\
&r_8\leq I(U;S_2)+H(S_1,S_2|U),\\
&r_9\leq I(U;S_1)+H(S_1,S_2|U),\\
&r_{10}\leq I(U;S_2)+H(S_1,S_2|U)
\end{align*}
for some $p_{U,(S_1,S_2)}=p_{U|(S_1,S_2)}p_{(S_1,S_2)}$ with $|\mathcal{U}|\leq|\mathcal{S}|+2$, and $\mathcal{R}(p_X, p_{Y_1,Y_2|X})$ is defined in Appendix \ref{app:TheoremI}. On the other hand, (\ref{eq:degenerate})  can be written equivalently as
\begin{align}
(H(S_0),H(S_1|S_0),H(S_2|S_0))\in\kappa\mathcal{C}_{\textsf{out}}(p_X,p_{Y_1,Y_2|X})\label{eq:degeneeratee2}
\end{align}
for some $p_X$. Therefore, it suffices to show that (\ref{eq:degeneeratee2}) implies (\ref{eq:ratecomp}) when $S_1\leftrightarrow S_0\leftrightarrow S_2$ form a Markov chain. Throughout the proof we assume $p_X$ is fixed.

It is clear that both $\mathcal{R}(p_{(S_1,S_2)})$ and $\mathcal{R}(p_X, p_{Y_1,Y_2|X})$ are closed convex sets. Let $\lambda_1,\cdots,\lambda_{10}$ be arbitrary non-negative numbers. We have
\begin{align}
&\max\limits_{(r_1,\cdots,r_{10})\in\mathcal{R}(p_{(S_1,S_2)})}\sum\limits_{i=1}^{10}\lambda_ir_i\nonumber\\
&=\max\limits_{p_{U|(S_1,S_2)}}\lambda_1I(U;S_1)+\lambda_2I(U;S_2)+\lambda_3 H(S_1)\nonumber\\
&\hspace{0.64in}+\lambda_4 H(S_2)+\lambda_5[ I(U;S_1)+H(S_2|U)]\nonumber\\
&\hspace{0.64in}+\lambda_6[I(U;S_2)+H(S_1|U)]\nonumber\\
&\hspace{0.64in}+\lambda_7[I(U;S_1)+H(S_1,S_2|U)]\nonumber\\
&\hspace{0.64in}+\lambda_8[I(U;S_2)+H(S_1,S_2|U)]\nonumber\\
&\hspace{0.64in}+\lambda_9[I(U;S_1)+H(S_1,S_2|U)]\nonumber\\
&\hspace{0.64in}+\lambda_{10}[I(U;S_2)+H(S_1,S_2|U)]\label{eq:max}\\
&=\max\limits_{p_{U|(S_1,S_2)}}(\lambda_1+\lambda_3+\lambda_5+\lambda_7+\lambda_9)H(S_1)\nonumber\\
&\hspace{0.64in}+(\lambda_2+\lambda_4+\lambda_6+\lambda_8+\lambda_{10})H(S_2)\nonumber\\
&\hspace{0.64in}-(\lambda_1+\lambda_5-\lambda_6+\lambda_7+\lambda_9)H(S_1|U)\nonumber\\
&\hspace{0.64in}-(\lambda_2-\lambda_5+\lambda_6+\lambda_8+\lambda_{10})H(S_2|U)\nonumber\\
&\hspace{0.64in}+(\lambda_7+\lambda_8+\lambda_9+\lambda_{10})H(S_1,S_2|U)\nonumber\\
&=\max\limits_{p_{U|(S_1,S_2)}}(\lambda_1+\lambda_3+\lambda_5+\lambda_7+\lambda_9)H(S_1)\nonumber\\
&\hspace{0.64in}+(\lambda_2+\lambda_4+\lambda_6+\lambda_8+\lambda_{10})H(S_2)\nonumber\\
&\hspace{0.64in}-(\lambda_1+\lambda_2)H(S_0|U)\nonumber\\
&\hspace{0.64in}-(\lambda_1+\lambda_5-\lambda_6+\lambda_7+\lambda_9)H(S_1|S_0,U)\nonumber\\
&\hspace{0.64in}-(\lambda_2-\lambda_5+\lambda_6+\lambda_8+\lambda_{10})H(S_2|S_0,U)\nonumber\\
&\hspace{0.64in}+(\lambda_7+\lambda_8+\lambda_9+\lambda_{10})H(S_1,S_2|S_0,U)\nonumber\\
&\leq\max\limits_{p_{U|(S_1,S_2)}}(\lambda_1+\lambda_3+\lambda_5+\lambda_7+\lambda_9)H(S_1)\nonumber\\
&\hspace{0.64in}+(\lambda_2+\lambda_4+\lambda_6+\lambda_8+\lambda_{10})H(S_2)\nonumber\\
&\hspace{0.64in}-(\lambda_1+\lambda_2)H(S_0|U)\nonumber\\
&\hspace{0.64in}-(\lambda_1+\lambda_5-\lambda_6-\lambda_8-\lambda_{10})H(S_1|S_0,U)\nonumber\\
&\hspace{0.64in}-(\lambda_2-\lambda_5+\lambda_6-\lambda_7-\lambda_{9})H(S_2|S_0,U),\label{eq:maxrelax}
\end{align}
where the last inequality follows from the fact that
\begin{align}
H(S_1,S_2|S_0,U)\leq H(S_1|S_0,U)+H(S_2|S_0,U).\label{eq:inequality}
\end{align}
Let $a=\lambda_1+\lambda_5-\lambda_6-\lambda_8-\lambda_{10}$ and $b=\lambda_2-\lambda_5+\lambda_6-\lambda_7-\lambda_{9}$. Consider the following four possible cases.
\begin{enumerate}
\item $a\leq 0$ and $b\leq0$: The maximum value of (\ref{eq:maxrelax}) is attained when $U=S_0$.

\item $a\geq 0$ and $b\leq 0$:  The maximum value of (\ref{eq:maxrelax}) is attained when $U=S_1$.

\item $a\leq 0$ and $b\geq 0$:  The maximum value of (\ref{eq:maxrelax}) is attained when $U=S_2$.

\item $a\geq 0$ and $b\geq 0$:  The maximum value of (\ref{eq:maxrelax}) is attained when $U=(S_1,S_2)$.
\end{enumerate}
It is clear that the equality holds in (\ref{eq:inequality}) for the following four choices of $U$:
\begin{enumerate}
\item $U=S_0$,

\item $U=S_1$,

\item $U=S_2$,

\item $U=(S_1,S_2)$.
\end{enumerate}
Therefore, the maximum value of (\ref{eq:max}) is also attained by one of these four choices of $U$; as a consequence, for  the necessary condition in Theorem \ref{thm:losslessnew}, there is no loss of generality in restricting $U$ to such choices. Note that (\ref{eq:degeneeratee2}) can be expressed alternatively as
\begin{align}
&H(S_0)\leq\kappa\min\{I(V^*_0;Y_1),I(V^*_0;Y_2)\},\label{eq:tver1}\\
&H(S_1)\leq\kappa[\min\{I(V^*_0;Y_1),I(V^*_0;Y_2)\}+I(V^*_1;Y_1|V^*_0)],\label{eq:tver2}\\
&H(S_2)\leq\kappa[\min\{I(V^*_0;Y_1),I(V^*_0;Y_2)\}+I(V^*_2;Y_2|V^*_0)],\label{eq:tver3}\\
&H(S_1,S_2)\leq\kappa[\min\{I(V^*_0;Y_1),I(V^*_0;Y_2)\}+I(V^*_1;Y_1|V^*_0)+I(X;Y_2|V^*_0,V^*_1)],\label{eq:tver4}\\
&H(S_1,S_2)\leq\kappa[\min\{I(V^*_0;Y_1),I(V^*_0;Y_2)\}+I(V^*_2;Y_2|V^*_0)+I(X;Y_1|V^*_0,V^*_2)]\label{eq:tver5}
\end{align}
for some $p_{V^*_0,V^*_1,V^*_2,X,Y_1,Y_2}=p_{V^*_0,V^*_1,V^*_2|X}p_Xp_{Y_1,Y_2|X}$. Setting $U=S_0$ in Theorem \ref{thm:losslessnew} yields the same set of constraints. On the other hand, when $U=S_1$, the necessary condition in Theorem \ref{thm:losslessnew} can be written as
\begin{align}
&H(S_0)\leq\kappa I(V_0;Y_2),\label{eq:ver1}\\
&H(S_1)\leq\kappa I(V_0;Y_1),\label{eq:ver2}\\
&H(S_2)\leq\kappa I(V_0,V_2;Y_2),\label{eq:ver3}\\
&H(S_2)\leq\kappa[I(V_0;Y_2)+I(V_1;Y_1|V_0)+I(X;Y_2|V_0,V_1)],\label{eq:ver4}\\
&H(S_1,S_2)\leq\kappa[I(V_0;Y_1)+I(V_2;Y_2|V_0)],\label{eq:ver5}\\
&H(S_1,S_2)\leq\kappa[I(V_0,V_1;Y_1)+I(X;Y_2|V_0,V_1)]\label{eq:ver6}
\end{align}
for some $p_{V_0,V_1,V_2,X,Y_1,Y_2}=p_{V_0,V_1,V_2|X}p_Xp_{Y_1,Y_2|X}$. By choosing  $V_0=V_1=(V^*_0,V^*_1)$ and $V_2=X$, we can see that
(\ref{eq:tver1})$\Rightarrow$(\ref{eq:ver1}), (\ref{eq:tver2})$\Rightarrow$(\ref{eq:ver2}), (\ref{eq:tver3})$\Rightarrow$(\ref{eq:ver3}), (\ref{eq:tver3})$\Rightarrow$(\ref{eq:ver4}), (\ref{eq:tver4})$\Rightarrow$(\ref{eq:ver5}), and (\ref{eq:tver4})$\Rightarrow$(\ref{eq:ver6}). The case $U=S_2$ follows by symmetry. When $U=(S_1,S_2)$,  the necessary condition in Theorem \ref{thm:losslessnew} can be written as
\begin{align}
&H(S_1)\leq\kappa I(V_0;Y_1),\label{eq:s1s2a}\\
&H(S_2)\leq\kappa I(V_0;Y_2),\label{eq:s1s2b}
\end{align}
for some $p_{V_0,X,Y_1,Y_2}=p_{V_0|X}p_Xp_{Y_1,Y_2|X}$. By choosing $V_0=X$, we can see that (\ref{eq:tver2})$\Rightarrow$(\ref{eq:s1s2a}) and (\ref{eq:tver3})$\Rightarrow$(\ref{eq:s1s2b}).
 Hence, (\ref{eq:ratecomp}) is indeed implied by (\ref{eq:degeneeratee2}) when $S_1\leftrightarrow S_0\leftrightarrow S_2$ form a Markov chain.
This completes the proof of Theorem \ref{thm:degenerate}.

\section{Proof of Lemma \ref{lem:capacitycomparison}}\label{app:capacitycomparison}

Let $(\hat{S}^m_{1},\hat{S}^m_{2})$ be jointly distributed with $S^m$ according to
\begin{align*}
p_{S^m}(s^m)p_{\hat{S}^m_1,\hat{S}^m_2|S^m}(\hat{s}^m_1,\hat{s}^m_2|s^m),
\end{align*}
where $p_{S^m}(s^m)=\prod_{t=1}^mp_S(s(t))$. We assume that $p_{\hat{S}^m_1,\hat{S}^m_2|S^m}$ is degraded with respect to $p_{Y^{\rho m}_1,Y^{\rho m}_2|X^m}$, where $p_{Y^{\rho m}_1,Y^{\rho m}_2|X^m}(y^{\rho m}_1,y^{\rho m}_2|x^m)=\prod_{q=1}^{\rho m}p_{Y_1,Y_2|X}(y_1(q),y_2(q)|x(q))$. As a consequence,
\begin{align}
\rho\mathcal{C}(p_{\hat{S}^m_1,\hat{S}^m_2|S^m})\subseteq\mathcal{C}(p_{Y^{\rho m}_1,Y^{\rho m}_2|X^m})=\rho m\mathcal{C}(p_{Y_1,Y_2|X}).\label{eq:capcomb1}
\end{align}
Let $\tilde{\mathcal{C}}_{\textsf{in}}(p_{S^m},p_{\hat{S}^m_1,\hat{S}^m_2|S^m})$ denote the set of $(R_0,R_1,R_2)\in\mathbb{R}^3_+$ satisfying
\begin{align*}
&R_0\leq\min\{I(U^m_{0};\hat{S}^m_{1}),I(U^m_{0};\hat{S}^m_{2})\},\\
&R_0+R_i\leq I(U^m_{0},U^m_{i};\hat{S}^m_{i}),\quad i=1,2,\\
&R_0+R_1+R_2\leq\min\{I(U^m_{0};\hat{S}^m_{1}),I(U^m_{0};\hat{S}^m_{2})\}+I(U^m_{1};\hat{S}^m_{1}|U^m_{0})+I(U^m_{2};\hat{S}^m_{2}|U^m_{0})-I(U^m_{1};U^m_{2}|U^m_{0})
\end{align*}
for some $(U^m_{0},U^m_{1},U^m_{2})$ be jointly distributed with $(S^m,\hat{S}^m_{1},\hat{S}^m_{2})$ such that $(U^m_{0},U^m_{1},U^m_{2})\leftrightarrow S^m\leftrightarrow(\hat{S}^m_{1},\hat{S}^m_{2})$ form a Markov chain,
and $(U_0(t),U_1(t),U_2(t),S(t))$, $t=1,\cdots,m$, are independent and identically distributed. It is clear that
\begin{align}
\tilde{\mathcal{C}}_{\textsf{in}}(p_{S^m},p_{\hat{S}^m_1,\hat{S}^m_2|S^m})\subseteq\mathcal{C}_{\textsf{in}}(p_{S^m},p_{\hat{S}^m_1,\hat{S}^m_2|S^m})\subseteq\mathcal{C}(p_{\hat{S}^m_1,\hat{S}^m_2|S^m}).\label{eq:capcomb2}
\end{align}
Let $T$ be a random variable independent of $(U^m_{0,1},U^m_{1,1},U^m_{2,1},S^m_1,\hat{S}^m_{1,1},\hat{S}^m_{2,1})$ and uniformly distributed over $\{1,\cdots,m\}$. Define
\begin{align*}
&U_i=U_i(T),\quad i=0,1,2,\\
&S=S(T),\\
&\hat{S}_i=\hat{S}_i(T),\quad i=1,2.
\end{align*}
Note that
\begin{align*}
I(U^m_{0};\hat{S}^m_{i})&=\sum\limits_{t=1}^mI(U_0(t);\hat{S}^m_{i}|U^{t-1}_{0})\\
&=\sum\limits_{t=1}^mI(U_0(t);\hat{S}^m_{i},U^{t-1}_{0})\\
&\geq\sum\limits_{t=1}^mI(U_0(t);\hat{S}_{i}(t))\\
&=mI(U_0(T);\hat{S}_i(T)|T)\\
&=mI(U_0(T);\hat{S}_i(T),T)\\
&\geq mI(U_0(T);\hat{S}_i(T))\\
&=mI(U_0;\hat{S}_i),\quad i=1,2;
\end{align*}
moreover,
\begin{align*}
I(U^m_{i};\hat{S}^m_{i}|U^m_{0})&=\sum\limits_{t=1}^mI(U_i(t);\hat{S}^m_{i}|U^m_{0},U^{t-1}_{i}),\\
&=\sum\limits_{t=1}^mI(U_i(t);\hat{S}^m_{i},U^{t-1}_{0},U^m_{0,t+1},U^{t-1}_{i}|U_{0}(t))\\
&\geq\sum\limits_{t=1}^mI(U_i(t);\hat{S}_i(t)|U_0(t))\\
&=mI(U_i(T);\hat{S}_i(T)|U_0(T),T)\\
&=mI(U_i(T);\hat{S}_i(T),T|U_0(T))\\
&\geq mI(U_i(T);\hat{S}_i(T)|U_0(T))\\
&=mI(U_i;\hat{S}_i|U_0),\quad i=1,2,
\end{align*}
and
\begin{align*}
I(U^m_{1};U^m_{2}|U^m_{0})=mI(U_1;U_2|U_0).
\end{align*}
Therefore, we have
\begin{align}
m\mathcal{C}_{\textsf{in}}(p_S, p_{\hat{S}_1,\hat{S}_2|S})\subseteq\tilde{\mathcal{C}}_{\textsf{in}}(p_{S^m},p_{\hat{S}^m_1,\hat{S}^m_2|S^m}).\label{eq:capcomb3}
\end{align}
Combining (\ref{eq:capcomb1}), (\ref{eq:capcomb2}), and (\ref{eq:capcomb3}) completes the proof of Lemma \ref{lem:capacitycomparison}.


\section*{Acknowledgment}

We would like to thank Prof. Chandra Nair and Prof. Amin Gohari for answering our numerous questions regarding broadcast channels. Furthermore, after seeing the conference version \cite{KC14} of the present paper, Prof. Gohari sent us a preprint of \cite{GA13b} (coauthored with  Venkat Anantharam), which contains, among other things, a computable characterization \cite[Corollary 2]{GA13b} of an earlier version of the Gohari-Anantharam outer bound \cite[Theorem 3]{GA08}, \cite[Theorem 3]{GA13a} as well as a necessary condition for the lossy source broadcast problem  \cite[Theorem 2]{GA13b}. 




\begin{thebibliography}{1}





\bibitem{HC87}
T. S. Han and M. H. M. Costa, ``Broadcast channels with arbitrarily
correlated sources," {\em IEEE Trans. Inf. Theory}, vol. IT-33, no. 5, pp.
641-650, Sep. 1987.

\bibitem{KN09}
G. Kramer and C. Nair, ``Comments on `Broadcast channels with
arbitrarily correlated sources'," in \emph{Proc. IEEE Int. Symp. Inform. Theory (ISIT)}, Seoul, Korea, Jun./Jul. 2009, pp. 2777-2779.

\bibitem{MK09}
P. Minero and Y.-H. Kim, ``Correlated sources over broadcast channels,"  in \emph{Proc. IEEE Int. Symp. Inform. Theory (ISIT)}, Seoul, Korea, Jun./Jul. 2009, pp. 2780-2784.

\bibitem{GA08}
A. A. Gohari and V. Anantharam,
``An outer bound to the admissible source region of broadcast channels with arbitrarily correlated sources and channel variations," in {\em Proc. 46th Annu. Allerton Conf. Commun., Control, Comput. (Allerton)}, Monticello, IL, Sep. 2008, pp. 301-308.

\bibitem{GA13a}
A. A. Gohari and V. Anantharam, ``Converses for discrete memoryless multiterminal networks," [Online]. Available: http://www.eecs.berkeley.edu/$\sim$ananth/2008+/ConversesForDMMN.pdf.




\bibitem{KLS09}
G. Kramer, Y. Liang, and S. Shamai (Shitz), ``Outer bounds on the admissible source region for broadcast channels with dependent sources," {\em Information Theory and Applications Workshop}, San Diego, CA, Feb. 8 - 13, 2009.

\bibitem{Nair11}
C. Nair, ``A note on outer bounds for broadcast channel," [Online]. Available: http://arxiv.org/abs/1101.0640v1.

\bibitem{GA13b}
A. A. Gohari and V. Anantharam, ``Infeasibility proof via information state," {\em IEEE Trans. Inf. Theory}, vol. 60, no. 10, pp.
5992-6004, Oct. 2014.

\bibitem{KC14}
K. Khezeli and J. Chen, ``An improved outer bound on the admissible source region for
broadcast channels with correlated sources," in \emph{Proc. IEEE Int. Symp. Inform. Theory (ISIT)}, Honolulu, HI, USA, Jun./Jul. 2014, pp. 466-470.




\bibitem{TDS11a}
C. Tian, S. Diggavi, and S. Shamai, ``Approximate characterizations for the Gaussian source broadcast distortion region,"  \emph{IEEE Trans. Inf. Theory}, vol. 57, no. 1, pp.
124-136, Jan. 2011.

\bibitem{CTDS14}
C. Tian, J. Chen, S. Diggavi, and S. Shamai (Shitz), ``Optimality and approximate optimality of
source-channel separation in networks," \emph{IEEE Trans. Inf. Theory}, vol. 60, no. 2, pp.
904-918, Feb. 2014.

\bibitem{GA12}
A. A. Gohari and V. Anantharam, ``Evaluation of Marton's inner bound
for the general broadcast channel," {\em IEEE Trans. Inf. Theory}, vol.~58, no. 2,
pp.~608-619, Feb. 2012.



\bibitem{Gohari13}
A. Gohari, private communication.



\bibitem{AGN13}
V. Anantharam, A. Gohari, and C. Nair, ``Improved cardinality bounds on the auxiliary
random variables in Marton's inner bound,"  in \emph{Proc. IEEE Int. Symp. Inform. Theory (ISIT)}, Istanbul, Turkey, Jul. 2013, pp. 1272-1276.


\bibitem{GP80}
S. I. Gelfand and M. S. Pinsker, ``Capacity of a broadcast channel with
one deterministic component," {\em Probl. Inform. Transm.}, vol. 16, no. 1,
pp. 17–25, Jan.-Mar. 1980.

\bibitem{CK81}
I. Csisz\'{a}r and J. K\"{o}rner, {\em Information Theory: Coding Theorems for
Discrete Memoryless Systems}. Akad\'{e}miai Kiad\'{o}: Budapest, 1981.



\bibitem{Marton79}
K. Marton, ``A coding theorem for the discrete memoryless broadcast channel," \emph{IEEE Trans.
Inf. Theory}, vol. IT-25,, no. 3, pp. 306-311, May 1979.




\bibitem{Nair13p}
C. Nair, private communication.

\bibitem{GK73}
P. G\'{a}cs and J. K\"{o}rner,  ``Common information is far less than mutual information," {\em Probl. Control Inf. Theory}, vol. 2, pp. 149-162, 1973.

\bibitem{Witsenhausen75}
H. S. Witsenhausen, ``On sequences of pairs of dependent random variables," {\em SIAM J. Appl. Math.}, vol. 28, no. 1, pp. 100-113, Jan. 1975






\bibitem{Han81}
T. S. Han, ``The capacity region for the deterministic broadcast channel with a common message," {\em IEEE Trans. Inf. Theory}, vol. IT-27, no. 1, pp.
122-125, Jan. 1981.








\bibitem{GGNY11}
Y. Geng, A. Gohari, C. Nair, and Y. Yu, ``The capacity region for two classes of product
broadcast channels," in \emph{Proc. IEEE Int. Symp. Inform. Theory (ISIT)}, Saint Petersburg, Russia, Jul./Aug. 2011, pp. 1544-1548.




\bibitem{KM77b}
J. K\"{o}rner and K. Marton, ``General broadcast channels with degraded message sets," \emph{IEEE Trans.
Inf. Theory}, vol. IT-23, no. 1, pp. 60-64, Jan. 1977.




\bibitem{KK08}
W. Kang and G. Kramer, ``Broadcast channel with degraded source
random variables and receiver side information," in \emph{Proc. IEEE Int. Symp. Inform. Theory (ISIT)}, Toronto, Canada,  Jul. 2008, pp. 1711-1715.




\bibitem{KM77}
J. K\"{o}rner and K. Marton, ``Comparison of two noisy channels," {\em Topics in Information Theory (Colloquia Mathematica Societatis J\'{a}nos Bolyai, Keszthely, Hungary, 1975)}, I. Csisz\'{a}r and P. Elias, Ed.  Amsterdam: North-Holland, 1977,  pp. 411-423.





\bibitem{EGK11}
A. El Gamal and Y.-H. Kim, {\em Network Information Theory}. Cambridge,
U.K.: Cambridge Univ. Press, 2011.


\bibitem{NEG07}
C. Nair and A. El Gamal, ``An outer bound to the capacity region of
the broadcast channel," {\em IEEE Trans. Inf. Theory}, vol.~53, no. 1,
pp.~350-355, Jan. 2007.




\end{thebibliography}
\end{document}